\documentclass[11pt]{article} 
\usepackage{fullpage}
\usepackage[font=footnotesize]{caption}
\usepackage{comment}
\usepackage{amsmath}
\usepackage{amsthm}
\usepackage{amssymb}
\usepackage{url}
\usepackage[final]{showkeys}
\usepackage{cleveref}
\usepackage{tikz}
\usepackage{lineno}

\newtheorem{theorem}{Theorem}
\newtheorem{lemma}{Lemma}

\newtheorem{conjecture}{Conjecture}
\newtheorem{claim}{Claim}

\newtheorem{definition}{Definition}

\newtheorem{fact}{Fact}
\newtheorem{observation}{Observation}
 
{\makeatletter
 \gdef\xxxmark{%
   \expandafter\ifx\csname @mpargs\endcsname\relax 
     \expandafter\ifx\csname @captype\endcsname\relax 
       \marginpar{xxx}
     \else
       xxx 
     \fi
   \else
     xxx 
   \fi}
 \gdef\xxx{\@ifnextchar[\xxx@lab\xxx@nolab}
 \long\gdef\xxx@lab[#1]#2{{\bf [\xxxmark #2 ---{\sc #1}]}}
 \long\gdef\xxx@nolab#1{{\bf [\xxxmark #1]}}
}

\newcommand{\m}[1]{\text{ }\left(\text{mod } #1 \right)}

\newcommand{\R}{\mathbb{R}}
\newcommand{\F}{\text{Frechet}}
\newcommand{\LCS}{\text{LCS}}
\newcommand{\poly}{\mbox{poly}}
\newcommand{\DTWD}{\mbox{DTWD}}

\newcommand{\EDIT}{\mbox{EDIT}}

\newcommand{\CG}{\mbox{CG}}
\newcommand{\AG}{\mbox{AG}}
\newcommand{\VG}{\mbox{VG}}
\newcommand{\WLCS}{\mbox{WLCS}}

\newcommand{\nn}{2^{N/2}}

\newcommand{\Ort}{\text{Orthogonal Vectors}}
\newcommand{\cnfsat}{\text{CNF-SAT}}
\newcommand{\MostOrt}{\text{Most-Orthogonal Vectors}}
\newcommand{\kMostOrt}{$k$\text{-Most-Orthogonal-Vectors}}
\newcommand{\kOrt}{$k$\text{-Orthogonal-Vectors}}

\newcommand{\seta}{\{\alpha_i\}_{i \in [n]}}
\newcommand{\setb}{\{\beta_i\}_{i \in [n]}}
\def \eps {\varepsilon}

\title{Quadratic-Time Hardness of LCS and other \\ Sequence Similarity Measures}

\date{}
\author{
	Amir Abboud\footnote{Supported by NSF Grant CCF-1417238, BSF Grant BSF:2012338 and a Stanford SOE Hoover Fellowship.} \\ Stanford University \\ \texttt{\small abboud@cs.stanford.edu}
	\and Arturs Backurs\\ MIT \\ \texttt{ \small backurs@mit.edu}
	\and Virginia Vassilevska Williams\footnotemark[\value{footnote}] \\ Stanford University \\ \texttt{\small virgi@cs.stanford.edu}
	}

\begin{document}


\begin{titlepage}
\clearpage\maketitle
\thispagestyle{empty}
\begin{abstract}

Two important similarity measures between sequences are the longest common subsequence (LCS) and the dynamic time warping distance (DTWD).
The computations of these measures for two given sequences are central tasks in a variety of applications.
Simple dynamic programming algorithms solve these tasks in $O(n^2)$ time, and despite an extensive amount of research, no algorithms with significantly better worst case upper bounds are known.

In this paper, we show that an $O(n^{2-\eps})$ time algorithm, for some $\eps>0$, for computing the LCS or the DTWD of two sequences of length $n$ over a constant size alphabet, refutes the popular Strong Exponential Time Hypothesis (SETH).
Moreover, we show that computing the LCS of $k$ strings over an alphabet of size $O(k)$ cannot be done in $O(n^{k-\eps})$ time, for any $\eps>0$, under SETH.
Finally, we also address the time complexity of approximating the DTWD of two strings in truly subquadratic time.

%
\end{abstract}
\end{titlepage}

\section{Introduction}
In many applications it is desirable to determine the similarity of two or more sequences of letters. 
The sequences could be English text, computer viruses, pointwise descriptions of points in the plane, or even proteins or DNA sequences.
Because of the large variety of applications, there are many notions of sequence similarity. Some of the most important notions are the Longest Common Subsequence (LCS), the Edit-Distance, the Dynamic Time Warping Distance (DTWD) and the Frechet distance measures. Considerable algorithmic research has gone into developing techniques to compute these measures of similarity. Unfortunately, even when the input consists of two strings, the time complexity of the problems is not well understood. There are classical algorithms that compute each of these measures in time that is roughly quadratic in the length of the strings, and this quadratic runtime is essentially the best known. A common technique to explain this quadratic bottleneck is to reduce the so called $3$SUM problem to the problems at hand. This approach has enjoyed a tremendous amount of success~\cite{overmars}. Nevertheless, there are no known reductions from $3$SUM to the above four sequence similarity problems. Two recent papers~\cite{Bring,edit_hardness} explained the quadratic bottleneck for Frechet distance and Edit-Distance by a reduction from CNF-SAT, thus showing that any polynomial improvement over the quadratic running time for these two problems would imply a breakthrough in SAT algorithms (refuting the Strong Exponential Time Hypothesis (SETH) that we define below).
A natural question is, can the same hypothesis explain the quadratic bottleneck for other sequence similarity measures such as DTWD and LCS? This paper answers this question in the affirmative, providing conditional lower bounds based on SETH for LCS and DTWD, along with other interesting results.


\paragraph{LCS.} Given two strings of $n$ symbols over some alphabet $\Sigma$, the LCS problem asks to compute the length of the longest sequence that appears as a subsequence in both input strings.
It is a very basic problem that we encounter in undergraduate-level computer science courses, with a classic $O(n^2)$ dynamic programming algorithm \cite{CLRS}.
LCS attracted an extensive amount of research, both for its mathematical simplicity and for its large number of important applications, including data comparison programs and bioinformatics.
In many of these applications, the size of $n$ makes the quadratic time algorithm impractical.
Despite a long list of improved algorithms for LCS and its variants in many different settings, e.g. \cite{Hir75, HS77} (see \cite{edita_vs_lcs} for a survey), the best algorithms on arbitrary strings are only slightly subquadratic and have an $O(n^2/\log^2{n})$ running time~\cite{masek1980faster} if the alphabet size is constant, and $O(n^2(\log\log{n})/\log^2{n})$ otherwise \cite{BF08,Grabo14}.

\paragraph{DTWD.}
Given two sequences of $n$ points $P_1$ and $P_2$, the \emph{dynamic time warping} distance between them is defined as the minimum, over
all monotone traversals of $P_1$ and $P_2$, of the sum over the stages of the traversal of the distance between the corresponding points at that stage (see the preliminaries for a formal definition).
When defined over symbols, the distance between two symbols is simply $0$ if they are equal and $1$ otherwise.
The DTWD problem asks to compute the score of the optimal traversal of two given sequences.
Note that if instead of taking the sum over all the stages of the traversal, we only take the maximum distance, we get the discrete Frechet distance between the sequences, a well known measure from computational geometry.

DTWD is an extremely useful similarity measure between temporal sequences which may vary in time or speed, and has long been used in speech recognition and more recently in countless data mining applications.
A simple dynamic programming algorithm solves DTWD in $O(n^2)$ time and is the best known in terms of worst-case running time, while many heuristics were designed in order to obtain faster runtimes in practice (see Wang et al. for a survey \cite{WDT+10}).

\paragraph{Hardness assumption.}
The Strong Exponential Time Hypothesis (SETH)~\cite{IP01,IPZ01} asserts that for any $\eps>0$ there is an integer $k>3$ such that $k$-SAT cannot be solved in $2^{(1-\eps)n}$ time.
Recently, SETH has been shown to imply many interesting lower bounds for polynomial time solvable problems~\cite{PW10, RV13,AV14,AVW14,Bring,edit_hardness}.
 We will base our results on the following conjecture, which is possibly more plausible than SETH: it is known to be implied by SETH, yet might still be true even if SETH turns out to be false. See Section~\ref{sec:sat} for a discussion.

\begin{conjecture}
\label{conj:MOV}
Given two sets of $n$ vectors $A,B$ in $\{0,1\}^d$ and an integer $r\geq 0$, there is no $\eps>0$ and an algorithm that can decide if there is a pair of vectors $a \in A, b \in B$ such that $\sum_{i=1}^{d} a_i b_i \leq r$, in $O(n^{2-\eps} \cdot \poly(d))$ time.
\end{conjecture}

\paragraph{Previous work.}
Out of the many recent SETH-based hardness results, most relevant to our work are the following three results concerning sequence similarity measures.

Abboud, Vassilevska Williams and Weimann \cite{AVW14} proved that a truly sub-quadratic algorithm\footnote{A truly (or strongly) sub-quadratic algorithm is an algorithm with $O(n^{2-\eps})$ running time, for some $\eps>0$.} for alignment problems like Local Alignment and Local-LCS refutes SETH.
However, the ``locality" of those measures was heavily used in the reductions, and the results did not imply any barrier for ``global" measures like LCS.

Bringmann~\cite{Bring} proved a similar lower bound for the computation of the Frechet distance problem.
As mentioned earlier, DTWD is equivalent to Frechet if we replace the ``max" with a ``sum".

Most recently, Backurs and Indyk~\cite{edit_hardness} proved a similar quadratic lower bound for Edit-Distance.
LCS and Edit-Distance are closely related.
A simple observation is that the computation of the LCS is equivalent to the computation of the Edit-Distance when only deletions and insertions are allowed, but no substitutions.
Thus, intuitively, LCS seems like an easier version of Edit-Distance, since it a solution has fewer degrees of freedom, and the lower bound for Edit-Distance does not immediately imply any hardness for LCS.

\subsection{Our results}
Our main result is to show that a truly sub-quadratic algorithm for LCS or DTWD refutes Conjecture~\ref{conj:MOV} (and SETH), and should therefore be considered beyond the reach of current algorithmic techniques, if not impossible.
Our results justify the use of sub-quadratic time heuristics and approximations in practice, and add two important problems to the list of SETH-hard problems!

\begin{theorem}
\label{thm:main}
If there is an $\eps>0$ such that either
\begin{itemize}
\item LCS 
over an alphabet of size $7$ can be computed in $O(n^{2-\eps})$ time, or
\item DTWD 
over symbols from an alphabet of size $5$ can be computed in $O(n^{2-\eps})$ time,
\end{itemize}
then Conjecture~\ref{conj:MOV} is false.
\end{theorem}

We note that the non-existence of $O(n^{2-\epsilon})$ algorithm for $\DTWD$ between two sequences of symbols
over an alphabet of size $5$ implies that there
is no $O(n^{2-\epsilon})$ time algorithm for $\DTWD$ between two sequences of points from $\ell_2^4$ ($4$-dimensional
Euclidean space). This follows because we can choose $5$ points in $4$-dimensional Euclidean space so that
any two points are at distance $1$ from each other, i.e., choose the vertices of a regular $4$-simplex.

Next, we consider the problem of computing the LCS of $k>2$ strings, which also is of great theoretical and practical interest.
A simple dynamic programming algorithm solves $k$-LCS in $O(n^k)$ time, and the problem is known to be NP-hard in general, even when the strings are binary~\cite{Maier78}.
When $k$ is a parameter, the problem is $W[1]$-hard, even over a fixed size alphabet, by a reduction from Clique~\cite{Piet03}. 
The parameters of the reduction imply that an $n^{o(k)}$ algorithm for $k$-LCS would refute ETH~\footnote{The exponential time hypothesis (ETH) is a weaker version of SETH: it asserts that there is some $\eps>0$ such that $3$SAT on $n$ variables requires $\Omega(2^{\eps n})$ time.}, and an algorithm with running time sufficiently faster than $O(n^{k/7})$ would imply a new algorithm for $k$-Clique.
However, no results ruling out $O(n^{k-1})$ or even $O(n^{k/2})$ upper bounds were known.

In this work, we prove that even a slight improvement over the dynamic programming algorithm is not possible under SETH when the alphabet is of size $O(k)$. 

\begin{theorem}
\label{thm:klcs}
If there is a constant $\eps>0$, an integer $k \geq 2$, and an algorithm that can solve $k$-LCS on strings of length $n$ over an alphabet of size $O(k)$ in $O(n^{k-\eps})$ time, then SETH is false.
\end{theorem}

A main question we leave open is whether the same lower bound holds when the alphabet size is a constant independent of $k$.
In Section~\ref{sec:klcs} we prove Theorem~\ref{thm:klcs} and make a step towards resolving the latter question by proving that a problem we call Local-$k$-LCS has such a tight $n^{k-o(1)}$ lower bound under Conjecture~\ref{conj:MOV} even when the alphabet size is $O(1)$.

Finally, we consider the possibility of truly sub-quadratic algorithms for approximating these similarity measures.
The LCS and Edit-Distance reductions do not imply any non-trivial hardness of approximation.
For Frechet in $2$-dimensional Euclidean space, Bringmann~\cite{Bring} was able to rule out truly sub-quadratic $1.0001$-approximation algorithms.
Here, we show that Bringmann's construction implies approximation hardness for DTWD and Frechet when the distance function between points is arbitrary, and is not required to satisfy the triangle inequality.
The details are presented in Section~\ref{sec:approx}.

\subsection{Technical contribution}
Our reductions build up on ideas from previous SETH-based hardness results for sequence alignment problems, and are most similar to the Edit-Distance reduction of \cite{edit_hardness}, with several new ideas in the constructions and the analysis.
As in previous reductions, we will need two kinds of gadgets: the vector or assignment gadgets, and the selection gadgets.
Two vector gadgets will be ``similar'' iff the two vectors satisfy the property we are interested in (we want to find a pair of
vectors that together satisfy some certain property).
The \emph{selection gadget} construction will make sure that the existence of a pair of ``similar'' vector-gadgets (i.e., the existence of a pair of vectors with the property), determines the overall similarity between the sequences. That is, if there is a pair of vectors
satisfying the property, the sequences are more ``similar'' than if there is non.
Typically, the vector-gadgets are easier to analyze, while the selection-gadgets might require very careful arguments.

There are multiple challenges in constructing and analyzing a reduction to LCS.
Our first main contribution was to prove a reduction from a weighted version of LCS (WLCS), in which different letters are more valuable than others in the optimal solution, to LCS.
Reducing problems to WLCS is a significantly easier and cleaner task than reducing to LCS.
Our second main contribution was in the analysis of the selection gadgets. 
The approach of \cite{edit_hardness} to analyze the selection gadgets involved a case-analysis which would have been extremely tedious if applied to LCS.
Instead, we use an inductive argument which decreases the number of cases significantly.

One way to show hardness of DTWD would be to show a reduction from Edit-Distance.
However, we were not able to show such a reduction in general.
Instead, we construct a mapping $f$ with the following property.
Given the hard instance of Edit-Distance, that were constructed in \cite{edit_hardness}, consisting of two sequences $x$ and $y$, we have that $\EDIT(x,y)=\DTWD(f(x),f(y))$. This requires carefully checking that this equality holds for particularly structures sequences.

\section{Preliminaries}

For an integer $n$, $[n]$ stands for $\{1,2,3,...,n\}$.

\subsection{Formal definitions of the similarity measures}

\begin{definition}[Longest Common Subsequence] \label{def:lcs}
For two sequences $P_1$ and $P_2$ of length $n$ over an alphabet $\Sigma$, the longest sequence $X$ that appears in both $P_1,P_2$ as a subsequence is the \emph{longest common subsequence} (LCS) of $P_1,P_2$ and we say that $LCS(P_1,P_2)=|X|$. The Longest Common Subsequence problem asks to output $LCS(P_1,P_2)$.
\end{definition}

\begin{definition}[Dynamic time warping distance]
For two sequences $x$ and $y$ of $n$ points from a set $\Sigma$ and a
distance function $d:\Sigma \times \Sigma \to \R^{0+}$,
the \emph{dynamic
time warping distance}, denoted by $\DTWD(x,y)$, is the minimum cost
of a (monotone) \emph{traversal} of $x$ and $y$. 

A traversal of the two sequences $x,y$ has the following form: We have two markers. Initially, one is located at the beginning of $x$, and the other is located at the
beginning of $y$.
 At every step, one or both of the
markers simultaneously move one point forward in their corresponding
sequences. 
At the end, both markers must be located at the last point of their corresponding
sequence.

 To determine the \emph{cost} of a traversal, we consider all the $O(n)$ steps of the traversal, and add up the following quantities to the final cost.
Let the configuration of a step be the pair of symbols $s$ and $t$ that the first and second markers are pointing at, respectively, then the contribution of this step to the final cost is $d(s,t)$.

The DTWD problems asks to output $DTWD(x,y)$.
\end{definition}

In particular, we will be interested in the following special case of DTWD.

\begin{definition}[$\DTWD$ over symbols]The DTWD problem over sequences of symbols, is the special case of DTWD in which the points come from an alphabet $\Sigma$ and the distance function is such that for any two symbols $s,t \in \Sigma$, $d(s,t)=1$ if $s \neq t$
and $d(s,t)=0$ otherwise.
\end{definition}

Besides LCS and DTWD which are central to this work, the following two important measures will be referred to in multiple places in the paper.

\begin{definition}[Edit-Distance] 
For any two sequences $x$ and $y$ over an alphabet $\Sigma$, the edit distance $\EDIT(x,y)$ is equal to the minimum number of symbol insertions, symbol deletions or symbol substitutions needed to transform $x$ into $y$. 
The Edit-Distance problem asks to output $\EDIT(x,y)$ for two given sequences $x,y$.
\end{definition}

\begin{definition}[The discrete Frechet distance]
The definition of the Frechet distance between two sequences of points is equivalent to the definition of the DTWD with
the following difference. Instead of defining the cost of a traversal to be the \emph{sum} of $d(s,t)$ for all the configurations of points $s$ and $t$ from the traversal, we define it to be the \emph{maximum} such distance $d(s,t)$. 
The Frechet problem asks to compute the minimum achievable cost of a traversal of two given sequences.
\end{definition}

\subsection{Satisfiability and Orthogonal Vectors}
\label{sec:sat}

%
%
%
%
%

To prove hardness based on Conjecture~\ref{conj:MOV} and therefore SETH, we will show reductions from the following vector-finding problems.

\begin{definition}[Orthogonal Vectors] \label{Def:Ort}
Given two lists $\{\alpha_i\}_{i\in [n]}$ and $\{\beta_i\}_{i\in [n]}$ of vectors $\alpha_i,\beta_i \in \{0,1\}^d$, is there a pair $\alpha_i,\beta_j$ that is orthogonal, $\sum_{h=1}^d \alpha_i[h]\cdot \beta_j[h] = 0$?
\end{definition}

This problem is known under many names and equivalent formulations, e.g. Batched Partial Match, Disjoint Pair, and Orthogonal Pair. 
Starting with the reduction of Williams~\cite{williams2005new}, this problem or variants of it have been used in every hardness result for a problem in P that is based on SETH, via the following theorem.

\begin{theorem}[Williams~\cite{williams2005new}]
If for some $\eps>0$, Orthogonal Vectors on $n$ vectors in $\{0,1\}^d$ for $d=O(\log{n})$ can be solved in $O(n^{2-\eps})$ time, then CNF-SAT on $n$ variables and $\poly(n)$ clauses can be solved in $O(2^{(1-\eps/2)n} poly(n))$ time, and SETH is false.
\end{theorem} 

The proof of this theorem is via the split-and-list technique and will follow from the proof of Lemma~\ref{lem:maxsat} below.
The following is a more general version of the Orthogonal Vectors problem.

\begin{definition}[Most-Orthogonal Vectors]
Given two lists $\{\alpha_i\}_{i\in [n]}$ and $\{\beta_i\}_{i\in [n]}$ of vectors 
$\alpha_i,\beta_i \in \{0,1\}^d$ and an integer $r \in \{0,\ldots,d\}$, 
is there a pair $\alpha_i,\beta_j$ that has inner product at most $r$, 
$\sum_{h=1}^d \alpha_i[h]\cdot \beta_j[h] \leq r$? 
We call any two vectors that satisfy this condition ($r$-)\emph{far}, and ($r$-)\emph{close} vectors otherwise.
\end{definition}

Clearly, an $O(n^{2-\eps})$ algorithm for Most-Orthogonal Vectors on $d$ dimensions implies a similar algorithm for Orthogonal Vectors, while the other direction might not be true.
In fact, while faster, mildly sub-quadratic algorithms are known for $\Ort$ when $d$ is polylogarithmic, with $O(n^2/\text{superpolylog($n$)})$ running times \cite{CIP02,ILPS14,AWY15}, we are not aware of any such algorithms for Most-Orthogonal Vectors. 

Lemma~\ref{lem:maxsat} below shows that such algorithms would imply new $O(2^n/\text{superpoly($n$)})$ algorithms for MAX-CNF-SAT on a polynomial number of clauses.
While such upper bounds are known for CNF-SAT~\cite{AWY15,DH09}, to our knowledge, $o(2^n)$ upper bounds are known for MAX-CNF-SAT only when the number of clauses is linear in the number of variables \cite{DW06,CK04}.
Together with the fact that the reductions from $\MostOrt$ to LCS, DTWD and Edit-Distance incur only a polylogarithmic overhead, this implies that shaving a superpolylogarithmic factor over the quadratic running times for these problems might be difficult.
The possibility of such improvements for pattern matching problems like Edit-Distance was recently suggested by Williams \cite{ryan-apsp}, as another potential application of his breakthrough technique for All-Pairs-Shortest-Paths.

More importantly, Lemma~\ref{lem:maxsat} shows that refuting Conjecture~\ref{conj:MOV} implies an $O(2^{(1-\eps)n}\poly(n))$ algorithm for MAX-CNF-SAT and therefore refutes SETH.

\begin{lemma}
\label{lem:maxsat}
If Most-Orthogonal Vectors on $n$ vectors in $\{0,1\}^d$ can be solved in $T(n,d)$ time, then given a CNF formula on $n$ variables and $M$ clauses, we can compute the maximum number of satisfiable clauses (MAX-CNF-SAT), in $O(T(2^{n/2},M)\cdot \log{M})$ time.
\end{lemma}

\begin{proof}
Given a CNF formula on $n$ variables and $M$ clauses, split the variables into two sets of size $n/2$ and list all $2^{n/2}$ partial assignments to each set.
Define a vector $v(\alpha)$ for each partial assignment $\alpha$ which contains a $0$ at coordinate $j \in [M]$ if $\alpha$ sets any of the literals of the $j^{th}$ clause of the formula to true, and $1$ otherwise. 
In other words, it contains a $0$ if the partial assignment satisfies the clause and $1$ otherwise.
Now, observe that if $\alpha,\beta$ are a pair of partial assignments for the first and second set of variables, then the inner product of $v(\alpha)$ and $v(\beta)$ is equal to the number of clauses that the combined assignment $(\alpha , \beta)$ does not satisfy.
Therefore, to find the assignment that maximizes the number of satisfied clauses, it is enough to find a pair of partial assignments $\alpha,\beta$ such that the inner product of $v(\alpha),v(\beta)$ is minimized.
The latter can be easily reduced to $O(\log{M})$ calls to an oracle for Most-Orthogonal Vectors on $N=2^{n/2}$ vectors in $\{0,1\}^M$ with a standard binary search. 
\end{proof}

By the above discussion, a lower bound that is based on Most-Orthogonal Vectors can be considered stronger than one that is only based on SETH.

\section{Hardness for LCS}

In this section we provide evidence for the hardness of the Longest Common Subsequence problem, and prove the first item in Theorem~\ref{thm:main}.

As an intermediate step, we first show evidence that solving a more general version of the problem in strongly subquadratic time is impossible under Conjecture~\ref{conj:MOV}.

\begin{definition}[Weighted Longest Common Subsequence ($\WLCS$)]
For two sequences $P_1$ and $P_2$ of length $n$ over an 
alphabet $\Sigma$ and a weight function $w:\Sigma \to [K]$,  
let $X$ be the sequence that appears in both $P_1,P_2$ as a subsequence 
and maximizes the expression $W(X)=\sum_{i=1}^{|X|} w(x[i])$. 
We say that $X$ is the $\WLCS$ of $P_1,P_2$ and write $\WLCS(P_1,P_2)=W(X)$.
The Weighted Longest Common Subsequence problem asks to output $\WLCS(P_1,P_2)$.
\end{definition}

Note that a common subsequence $X$ of two sequences $P_1,P_2$ can be thought of as an alignment or a matching $A = \{ (a_i,b_i) \}_{i=1}^{|X|}$ between the two sequences, so that for all $i \in [|X|]: P_1[a_i]=P_2[b_i]$, and $a_1<\cdots<a_{|X|}$ and $b_1 < \cdots < b_{|X|}$.
Clearly, the weight $\sum_{i=1}^{|X|}P_1[a_i]=\sum_{i=1}^{|X|}P_2[b_i]$ 
of the matching $A$ correspond to the length $W(X)$ of the weighted length of the common subsequence $X$.

In our proofs, we will find useful the following relation between pairs of indices.
For a pair $(x,y)$ and a pair $(x',y')$ of indices 
we say that they are in \emph{conflict} or they \emph{cross} 
if $x<x'$ and $y>y'$ or $x>x'$ and $y<y'$.

\subsection{Reducing $\WLCS$ to $\LCS$}

The following simple reduction from $\WLCS$ to $\LCS$ gives a way to translate a lower bound for WLCS to a lower bound for LCS, and allows us to simplify our proofs.

\begin{lemma}
\label{lem:wlcs}
Computing the $\WLCS$ of two sequences of length $n$ over $\Sigma$ with 
weights $w: \Sigma \to [K]$ can be reduced to computing the $\LCS$ of 
two sequences of length $O(Kn)$ over $\Sigma$.
\end{lemma}

\begin{proof}
The reduction simply copies each symbol $\ell \in \Sigma$ in each of the sequences $w(\ell)$ times.
That is, we define a mapping $f$ from symbols in $\Sigma$ to sequences of length up to $K$ so that for any $\ell \in \Sigma$, $f(\ell) = [ \ell^{w(\ell)}] \in \Sigma^{w(\ell)}$.

For a sequence $P$ of length $n$ over $\Sigma$, let $f(P) = \bigcirc_{i=1}^{n} f(P[i])$. That is, replace the $i^{th}$ symbol $P[i]$ with the 
sequence $f(P[i])$ defined above.

Note that $|f(P)| \leq K |P|$ and the reduction follows from the next claim.

\begin{claim}
For any two sequences $P_1,P_2$ of length $n$ over $\Sigma$, the mapping $f$ satisfies: \[WLCS(P_1,P_2) = LCS(f(P_1),f(P_2)).\]
\end{claim}

\begin{proof}
For brevity of notation, we let $P_1'=f(P_1)$ and $P_2'=f(P_2)$.

First, observe that $WLCS(P_1,P_2) \leq LCS(P_1',P_2')$, since for any common subsequence $X$ of $P_1,P_2$, the sequence $f(X)$ is a common subsequence of $P_1',P_2'$ and has length $|f(X)|=\sum_{i=1}^{n} |f(X[i])| = \sum_{i=1}^n w(X[i]) = W(X)$.

In the remainder of this proof, we show that $WLCS(P_1,P_2) \geq LCS(P_1',P_2')$.
Let $X$ be the LCS of $P_1',P_2'$ and consider a corresponding matching $A$.
%

Let $x \in \{1,2\}$. We say that a symbol $\ell$ in $P_x'$ at index $i \leq Kn$ belongs to interval $I_x(i) \in [n]$, iff this symbol was generated when mapping $P_x[I_x(i)]$ to the subsequence $f(\ell)$. 
Moreover, we say that it is at index $J_x(i) \in [w(\ell)]$ in interval $I_x(i)$, iff it is the $J_x(i)^{th}$ symbol in that interval. 

We will go over the symbols $\ell \in \Sigma$ of the alphabet in an arbitrary order, and perform the following modifications to $X$ and the matching $A$ for each such symbol in turn.

Go over the indices $i$ of $P_1'$ that are matched in $A$ to some index $j$ of $P_2'$, and for which $P_1'[i]=\ell$, in increasing order.
Consider the intervals $I_1(i)$ and $I_2(j)$, both of which contain the symbol $\ell$, $w(\ell)$ times.
Throughout our scan, we maintain the invariant that: $i$ is the first index to be matched to the interval $I_2(j)$.

If $J_1(i)=J_2(j)=1$, and the next $w(\ell)-1$ pairs in our matching $A$ are matching the rest of the interval $I_1(i)$ to the interval $I_2(j)$, we do not need to modify anything, and we move on to the next index $i'$ that is not a part of this interval $I_1(i)$ and is matched to some index $j'$ - note that at this point, $i'$ satisfies the invariant, since it cannot also be matched to the interval $I_2(j)$ by the pigeonhole principal, and therefore $I_2(j')>I_2(j)$ and $i'$ is the first index to be matched to this interval.

Otherwise, we modify $A$ so that now the whole intervals $I_1(i)$ and $I_2(j)$ are matched to one another: for each $i',j'$ 
such that $I_1(i')=I_1(i), I_2(j')=I_2(j)$, and $J_1(i')=J_2(j')$, 
we add pair $(i',j')$ to the matching $A$, 
and remove any conflicting pairs from $A$.
We claim that we obtain a matching of at least the original size, since we add $w(\ell)$ pairs and we remove only up to $w(\ell)$ pairs.
To see this, note that for a pair $(x,y)$ to be in conflict with one of the pairs we added, it must be one of the following three types: (1) $I_1(x)=I_1(i)$ and $I_2(y)=I_2(j)$, or (2) $I_1(x)=I_1(i)$ but $I_2(y)> I_2(j)$, or (3) $I_2(y) = I_2(j)$ but $I_1(x)>I_1(i)$. 
Here, we use the invariant to rule out pairs for which $I_1(x) < I_1(i)$ or $I_2(y)<I_2(j)$.
However, in any matching $A$, there cannot be both pairs of type (2) and pairs of type (3), since any such two pairs would cross.
Therefore, we conclude that all conflicting pairs either come from the interval $I_1(i)$ or they all come from the interval $I_2(j)$, and in any case, there are only $w(\ell)$ of them.
After this modification, we move on to the next index $i'$ that is not a part of this interval $I_1(i)$ and is matched (in the new matching $A$) to some index $j'$ - as before, this $i'$ satisfies the invariant.

After we are done with all these modifications, we end up with a matching $A$ of size at least $|X|$ in which complete intervals are aligned to each other.
Now, we can define a matching $A'$ between $P_1$ and $P_2$ that contains all pairs $(I_1(i),I_2(j))$ for which $(i,j) \in A$.
In words, we contract the intervals of $P_1',P_2'$ to the original symbols of $P_1,P_2$.
Finally, $A'$ corresponds to a common subsequence $X'$ of $P_1,P_2$, and $W(X')=|A|\geq |X|$ since each matched interval corresponds to some symbol $\ell$ and contributes $w(\ell)$ matches to $A$ and a single match of weight $w(\ell)$ to $A'$.
\end{proof}
\end{proof}

\subsection{Reducing $\MostOrt$  to $\LCS$ }

We are now ready to present our main reduction, proving our hardness result for LCS. 

\begin{theorem}
$\MostOrt$ on two lists $\{\alpha_i\}_{i\in [n]}$ and $\{\beta_i\}_{i\in [n]}$ of $n$ binary vectors in $d$ dimensions
($\alpha_i,\beta_i \in \{0,1\}^d$)
can be reduced to $\LCS$ problem on two sequences of length $n \cdot d^{O(1)}$ over an 
alphabet of size $7$.
\end{theorem}

\begin{proof}
We will proceed in two steps. First, we will show that $\WLCS$ is at least as hard as the $\MostOrt$ problem.
Second, given that the symbols in the constructed $\WLCS$ instance will have small weights,
an application of Lemma~\ref{lem:wlcs} will allow as to conclude that $\LCS$ is at least as hard as the
$\MostOrt$ problem.
Our alphabet will be $\Sigma=\{0,1,2,3,4,5,6\}$.


We start with the reduction to $\WLCS$.
Let $\alpha, \beta$ denote two vectors from the $\MostOrt$ instance, from the first and the second set, respectively.

We construct our \emph{coordinate gadgets} as follows. For $i \in [d]$ we define,

\[
CG_1(\alpha,i) = \begin{cases}
               5465 & \text{if $\alpha[i]=0$}\\
               545 & \text{otherwise}\\
            \end{cases}
            \]
            
            \[
            CG_2(\beta,i) = \begin{cases}
               5645 & \text{if $\beta[i]=0$}\\
               565 & \text{otherwise}\\
            \end{cases}
\]

Setting the weight function so that $w(4)=w(6) =1, w(5) = X = 100 d$.

These gadgets satisfy the following equalities:
\[
WLCS(CG_1(\alpha,i), CG_2(\beta,i)) = \begin{cases}
               2X+1 & \text{if $\alpha[i] \cdot \beta[i]=0$}\\
               2X & \text{otherwise}\\
            \end{cases}
            \]
Now, we define the \emph{vector gadgets} as a concatenation of the coordinate gadgets. 
Let $R_1(\alpha) = \bigcirc_{i=1}^d CG_1(\alpha,i)$ and $R_2(\beta) = \bigcirc_{i=1}^d CG_2(\beta,i)$.
\[
VG_1(\alpha) =  1 \circ R_1(\alpha) 
\]
\[
VG_2(\beta) =   R_2(\beta) \circ 1
\]

The weight of the symbol $1$ is $w(1) = A = (r+1)2X+(d-(r+1))(2X+1)$.
It is now easy to prove the following claims.

\begin{claim} If two vectors $\alpha,\beta$, are $r$-far, then:
\[
WLCS(VG_1(\alpha), VG_2(\beta)) \geq A+1= r\cdot 2X+(d-r)(2X+1).
\]
\end{claim}

\begin{proof}
For each $i \in [d]$, match $CG_2(\beta,i)$ to $CG_1(\alpha,i)$ optimally to get a weight at least $A+1=r \cdot 2X+(d-r)(2X+1)$.
\end{proof}

\begin{claim} \label{rclose}
 If two vectors $\alpha,\beta$, are $r$-close, then:
\[
\WLCS(VG_1(\alpha), VG_2(\beta)) = A. \]
\end{claim}

\begin{proof}
$\WLCS(VG_1(\alpha), VG_2(\beta)) \geq A$ is true because we can match the $1$ symbols, which gives cost $A$.

Now we prove that $\WLCS(VG_1(\alpha), VG_2(\beta)) \leq A$.
If we match the $1$ symbols, then we cannot match any other symbols and the inequality is true. Thus, we assume now that the $1$ symbols are not matched.

Now we can check that, if there is a $5$ symbol in $\VG_1(\alpha)$ or $\VG_2(\beta)$ 
that is not matched to a $5$ symbol, then we cannot achieve weight $A$ even
if we match all the other symbols (except for the $1$ symbols). Therefore, we assume
that all the $5$ symbols are matched. The required inequality follows 
from the fact that there are at least $r+1$ coordinates
where $\alpha$ and $\beta$ both have $1$ (the vectors are $r$-close), and the construction of the coordinate gadgets.
\end{proof}

Finally, we combine the vector gadgets into two sequences.
Let $VG_1'(\alpha) = 0 \circ VG_1(\alpha) \circ 2$ and $VG_2' (\beta) = 0 \circ VG_2(\beta) \circ 2 \circ 3$. 
Let $f$ be a dummy vector of length $d$ that is all $1$.
\[
P_1 = 3^{|P_2|} \circ \bigcirc_{i=1}^n VG_1'(\alpha_i) \circ 3^{|P_2|}
\]
\[
P_2 = 3 \circ \bigcirc_{i=1}^{n-1} VG_2'(f) \circ \bigcirc_{i=1}^{n} VG_2'(\beta_i) \circ \bigcirc_{i=1}^{n-1} VG_2'(f) \]
And set the weights so that $w(3) = B = A^2$ and $w(0)=w(2)= C = B^2$.

Let $E_U = 2C + A$, and $E_G = n \cdot E_U + 2n \cdot B$.

The following two lemmas prove that there is a gap in the $\WLCS$ of our 
two sequences when there is a pair of vectors that are $r$-far as 
opposed to when there is none.

\begin{lemma}
If there is a pair of vectors that are $r$-far, 
then $\WLCS(P_1,P_2) \geq E_G + 1$.
\end{lemma}

\begin{proof}
Let $i,j$ be such that $\alpha_i,\beta_j$ are $r$-far.
Match $VG_1'(\alpha_i)$ and $VG_2'(\beta_j)$ to get a weight of 
at least $2C+r \cdot 2X+(d-r)(2X+1) \geq E_U +1$.
Match the $i-1$ vector gadgets to the left of $VG_1'(\alpha_i)$ to the $i-1$ vector gadgets immediately to the left of $VG_2'(\beta_j)$, and similarly, match the $n-i$ gadgets to the right. 
The total additional weight we get is at least $(n-1)\cdot E_U$.
Finally, note that after the above matches, only $(n-1)$ out of the $(3n-1)$ $3$-symbols in $P_2$ are surrounded by matched symbols.
The remaining $2n$ $3$-symbols can be matched, 
giving an additional weight of $2n\cdot B$.
The total weight is at least $E_U +1 + (n-1)\cdot E_U + 2n\cdot B = E_G +1$.
\end{proof}

\begin{lemma}
\label{lem:LCSmain}
If there is no pair of vectors that are $r$-far, then $\WLCS(P_1,P_2) \leq E_G$.
\end{lemma}

\begin{proof}
The main part of the proof will be dedicated to showing that if the $n$ vector gadgets in $P_1$ are matched to a substring of $n'$ vector gadgets from $P_2$, then $n'$ must be equal to $n$. 
This will follow since: if $n'<n$, then at least one of the $0$/$2$ symbols in $P_1$ will remain unmatched, and, if $n'>n$, then less than $2n$ of the $3$ symbols in $P_2$ can be matched. 
The large weights we gave $0$/$2$ and $3$ make this impossible in an optimal matching.
It will be easy to see that in any matching in which $n=n'$, the total weight is at most $E_G$.

Now, we introduce some notation.
Let $L \leq L'$ and define $W(L,L')$ to be the optimal score of matching two sequence $T,T'$ where $T$ is composed of $L$ vector gadgets $VG_1'(\alpha)$ and $T'$ is composed of $L'$ vector gadgets $VG_2'(\beta)$, where no pair $\alpha,\beta$ are $r$-far.
Define $W_0(L,L')$ similarly, except that we restrict the matchings so that all $0$ or $2$ symbols in $T$ (the shorter sequence) must be matched.
In the following two claims we prove an upper bound on $W(L,L')$, via an upper bound on $W_0(L,L')$.

\begin{claim}
\label{cl:0}
For any integers $1\leq L \leq L'$, we can upper bound $W_0(L,L') \leq L \cdot E_U + (L'-L) \cdot (B-1)$.
\end{claim}

\begin{proof}
Let $T,T'$ be two sequences with $L,L'$ vector gadgets, respectively. 
We will refer to these ``vector gadgets" as intervals.
Consider an optimal matching of $T$ and $T'$ in which all the $0$ and $2$ symbols of $T$ are matched, i.e., a matching that achieves weight $W_0(L,L')$ - we will upper bound its weight $E_F$ by $L \cdot E_U + (L'-L) \cdot (B-1)$.
Note that in such a matching, each interval of $T$ must be matched completely within one or more intervals of $T'$, and each interval of $T'$ has matches to at most one interval from $T$ (otherwise, it must be the case that some $0$ or $2$ symbol in $T$ is not matched).

Let $x$ be the number of intervals of $T$ that contribute at most $E_U$ to the weight of our optimal matching. 
Note that any of the $L-x$ other intervals must be matched to a substring of $T'$ that contains at least two intervals for the following reason. The $0$ and $2$ symbols of the interval of $T'$ must be matched, and, if the matching stays within a single interval of $T'$ and has more than $E_U$ weight, then we have a pair which is $r$-far because of Claim \ref{rclose}.
Thus, using the fact that there are only $L'$ intervals in $T'$, we get the condition, 
$$
x+2(L-x) \leq L'.
$$

We now give an upper bound on the weight of our matching, by summing the contributions of each interval of $T$:
there are $x$ intervals contributing $\leq E_U$ weight, and there are $(L-x)$ intervals matched to $T'$ with unbounded contribution, but we know that even if all the symbols of an interval are matched, it can contribute at most $E_B = 2C + A + d(2X+2)$.
Therefore, the total weight of the matching can be upper bounded by 
\[
E_F \leq (L-x) \cdot E_B + x \cdot E_U 
\]
We claim that no matter what $x$ is, as long as the above condition holds, this expression is less than $L \cdot E_U + (L'-L) \cdot (B-1)$.

To maximize this expression, we choose the smallest possible $x$ that satisfies the above condition, since $E_B > E_U$, which implies that $x = \max\{ 0, 2L-L' \}$.
 A key inequality, which we will use multiple times in the proof, following from the fact that the $0$/$2$/$3$ symbols are much more important than the rest, is that $E_B < E_U + B - 1$, which follows since $E_B-E_U < A+d(2X+2)  <  1000d^2 < B$.

First, consider the case where $L\leq L'/2$, and therefore $x=0$,
which means that all the intervals of $T$ might be fully matched.
 Using that $E_B < E_U+B-1$ and that $L'-L \geq L'/2 \geq L$, we get the desired upper bound:
$$
E_F \leq L \cdot E_B \leq L \cdot (E_U + B-1) \leq L \cdot E_U + (L'-L) \cdot (B-1).
$$

Now, assume that $L > L'/2$, and therefore $x = 2L-L'$.
In this case, when setting $x$ as small as possible, the upper bound becomes:
\[
E_F 
\leq (L'-L)\cdot E_B + (2L-L') \cdot E_U 
= L \cdot E_U + (L'-L) \cdot (E_B - E_U),
\]
which is less than $L \cdot E_U + (L'-L) \cdot (B-1)$, since $E_B < E_U + B-1$.
\end{proof}

Next, we prove by induction that leaving  $0$/$2$ symbols in the shorter sequence unmatched will only worsen the weight of the optimal matching.

\begin{claim}
\label{cl:1}
For any integers $1\leq L \leq L'$, we can upper bound $W(L,L') \leq L \cdot E_U + (L'-L) \cdot (B-1)$.
\end{claim}

\begin{proof}
We will prove by induction on $i \geq 2$ that: for all $L'\geq L \geq 1$ such that $L+L' \leq i$, $W(L,L') \leq L \cdot E_U + (L'-L) \cdot (B-1)$.

The base case is when $i=2$ and $L=L'=1$. Then $W(1,1) = E_U$ and we are done.

For the inductive step, assume that the statement is true for all $i' \leq i-1$ and we will prove it for $i$.
Let $L,L'$ be so that $1\leq L\leq L' $ and $L+L'=i$ and let $T,T'$ be sequences with $L,L'$ intervals (assignment gadgets), respectively.
Consider the optimal (unrestricted) matching of $T$ and $T'$, denote its weight by $E_F$. 
Our goal is to show that $E_F \leq L \cdot E_U + (L'-L) \cdot (B-1)$.

If every $0$/$2$ symbol in $T$ is matched then, by definition, the weight cannot be more than $W_0(L,L')$, and by Claim~\ref{cl:0} we are done.
Otherwise, consider the first unmatched $0$/$2$ symbol, call it $x$, and there are two cases.

\text{ \bf The $x= 0$ case:} 
If $x$ is the first $0$ in $T$, then the first $0$ in $T'$ must be matched to some $0$ after $x$ (otherwise we can add this pair to the matching without violating any other pairs) which implies that none of the symbols in the interval starting at $x$ can be matched, since such matches will be in conflict with the pair containing this first $0$.
Otherwise, consider the $2$ that appears right before $x$ and note that it must be matched to some $y=$ 2 in $T'$, by our choice of $x$ as the first unmatched $0$/$2$. 
Now, there are two possibilities: 
either there are no more intervals in $T'$ after $y$, or there is a $0$ right after $y$ in $T'$ that is matched to a $0$ in $T$ that is after $x$ (from a later interval in $T$).
Note that in either case, the interval starting at $x$ (and ending at the $2$ after it) is completely unmatched in our matching.
Therefore, in this case, we let ${T_1}$ be the sequence with $(L-1)$ intervals which is obtained from $T$ by removing the interval starting at $x$. 
The weight of our matching will not change if we look at it as a matching between $T'$ and ${T_1}$ instead of $T$, which implies that $E_F \leq W(L-1,L')$.
Using our inductive hypothesis we conclude that $E_F \leq (L-1) \cdot E_U + (L'-L+1) \cdot (B-1) \leq L \cdot E_U + (L'-L) \cdot (B-1)$, since $E_U > B$, and we are done.

\text{ \bf The $x= 2$ case:} The $0$ at the start of $x$'s interval must have been matched to some $y = 0$. 
Let $z$ be the $2$ at the end of $y$'s interval.
Note that $z$ must be matched to some $w= 2$ in $T$ after $x$, since otherwise, we can add the pair $(x,z)$ to the matching, gaining a cost of $C$, and the only possible conflicts we would create will be with pairs containing a symbol inside the $y \to z$ interval or inside $x$'s interval, and if we remove all such pairs, we would lose at most $(A+d(2X+2))$ which is much less than the gain of $C$ - implying that our matching could not have been optimal.
Therefore, there are $c \geq 2$ intervals in $T$ that are matched to a single interval in $T'$: 
all the intervals starting at the $0$ right before $x$ and ending at $w$ are matched to the $y \to z$ interval.
Let $T_1$ be the sequence obtained from $T$ by removing all these $c$ intervals and let $T_2$ be the sequence obtained from $T'$ by removing the $y \to z$ interval.
Our matching can be split into two parts: a matching between $T_1$ and $T_2$, and the matching of the $y \to z$ interval to the removed interval.
The contribution of the latter part to the weight of the matching can be at most the weight of all the symbols in an interval, which is $E_B$.
By the inductive hypothesis, we know that any matching of $T_1$ and $T_2$ can have weight at most $W(L-c,L'-1)\leq (L-c) \cdot E_U + (L'-1-L+c) \cdot (B-1)$.
Summing up the two bounds on the contributions, we get that the total weight of the matching is at most:
\[
E_F \leq E_B + (L-c) \cdot E_U + (L'-L+c-1) \cdot (B-1)
\leq L \cdot E_U + (L'-L) \cdot (B-1) + (c-1) \cdot (B-1) + E_B - c \cdot E_U 
\]
However, note that $E_B < 1.1 E_U $ and that $(c-1.1) E_U > 10(c-1.1) B > (c-1)B$, which implies that $E_F$ can be upper bounded by $L \cdot E_U + (L'-L) \cdot (B-1)$, and we are done.
\end{proof}

We are now ready to complete the proof of the Lemma.
Consider the optimal matching of $P_1$ and $P_2$.
Let $x$ and $y$ be the first and last $3$ symbols in $P_2$ that are not matched, respectively.
Note that there cannot be any matched $3$ symbols between $x$ and $y$, since otherwise we could  match either $x$ or $y$ and gain extra weight without incurring any loss. 
Moreover, note that $x$ cannot be the first symbol in $P_2$ and $y$ cannot be the last one, since those must be matched in an optimal alignment.
The substring between the 3 preceding $x$, and the 3 following $y$, contains $n'$ intervals (vector gadgets) for some $ 1 \leq n' \leq 3n-2$.
If all the 3's are matched, we let $n'=1$, and focus on the only interval (vector gadget) of $P_2$ that has matched non-$3$-symbols.

We can now bound the total weight of the matching by the sum of the maximum possible contribution of these $n'$ intervals, and the contribution of the rest of $P_2$.
The substring before and including the $3$ symbol preceding $x$ and the substring after and including the $3$ symbol following $y$ can only contribute $3$'s to the matching, and they contain exactly $(3n-1 - (n'-1))$ such $3$ symbols, giving a contribution of $(3n-n')\cdot B$.
To bound the contribution of the $n'$ intervals, we use Claim~\ref{cl:1}:
since no $3$ symbols are matched in this part, we can ``remove" those symbols for the analysis, to obtain two sequences $T,T'$ composed of $n,n'$ vector gadgets, respectively, in which no pair is $r$-far.
The contribution of the $T,T'$ part, depends on $n,n'$:

If $n' \leq n$, then by Claim~\ref{cl:1}, when setting $L=n', L'=n$, the contribution is at most $(n' \cdot E_U + (n-n')\cdot (B-1))$ and the total weight of our matching can be upper bounded by 
\[
(3n-n')\cdot B + (n' \cdot E_U + (n-n')\cdot (B-1)),
\]
which is maximized when $n'$ is as large as possible, since $E_U > (2B-1)$. Thus, setting $n'=n$, we get the upper bound: $(3n-n)\cdot B + n \cdot E_U = E_G$.

Otherwise, if $n' > n$, we apply Claim~\ref{cl:1} with $L=n, L'=n'$, and get that the contribution is at most $(n \cdot E_U + (n'-n)\cdot (B-1))$, and the total weight of our matching can be upper bounded by 
\[
(3n-n')\cdot B + (n \cdot E_U + (n'-n)\cdot (B-1)) 
= n \cdot E_U + 2n \cdot B - (n'-n)  < E_G.
\]
\end{proof}

To conclude our reduction, we note that the largest weight used in our weight function is polynomial in $d$, and therefore the reduction of Lemma~\ref{lem:wlcs} gives two unweighted sequences $f(P_1),f(P_2)$ of length $n\cdot d^{O(1)}$, for which the LCS equals the WLCS of our $P_1,P_2$.
\end{proof}

\section{Hardness for DTWD}

In this section, we complete the proof of Theorem~\ref{thm:main} by showing that a truly sub-quadratic algorithm for DTWD implies a truly sub-quadratic algorithm for the $\MostOrt$ problem.

We first show that we can modify the reduction from $\cnfsat$ to Edit-Distance from \cite{edit_hardness}
so that we get a reduction from $\MostOrt$ to Edit-Distance.
We will later use properties of the two sequences produced in this reduction, call them $P_1',P_2'$.
In particular, we will show that there is an easy transformation of $P_1'$ into a sequence $S_1$ and
of  $P_2'$ into a sequence $S_2$ so that $\EDIT(P_1',P_2')=\DTWD(S_1,S_2)$.
This will give the desired reduction from $\MostOrt$ to $\DTWD$.

\subsection{Reducing $\MostOrt$ to Edit-Distance}

Before showing the reduction from $\MostOrt$ to Edit-Distance, let us recast the reduction of \cite{edit_hardness} as a reduction from $\Ort$ instead of CNF-SAT.

	
	\paragraph{Reducing $\Ort$ to Edit-Distance.}
	Instead of having $\nn$ partial assignments for the first half of the variables
	and $\nn$ partial assignments for the second half of the variables, we
	have $n$ vectors in the first and the second set of vectors (we replace $\nn$ by $n$ in the argument).
	Instead of having $M$ clauses, we have $d$ coordinates for every vector
	(we replace $M$ by $d$ in the argument).

	Instead of having \emph{clause gadgets}, we have \emph{coordinate gadgets}.
	For a vector $\alpha$ from the first set of vectors $\seta$ and $j \in [d]$,
	we define a coordinate gadget,
	$$
		\CG_1(\alpha,j)=
		\begin{cases}
			0^{l_1}0^{l_0}1^{l_0}1^{l_0}1^{l_0}0^{l_1} & \text{ if }\alpha[j]=0,\\
			0^{l_1}0^{l_0}0^{l_0}0^{l_0}1^{l_0}0^{l_1} & \text{ otherwise.}
		\end{cases}
	$$
	For a vector $\beta$ from the second set of vectors $\setb$ and $j \in [d]$,
	$$
		\CG_2(\beta,j)=
		\begin{cases}
			0^{l_1}0^{l_0}0^{l_0}1^{l_0}1^{l_0}0^{l_1} & \text{ if }\beta[j]=0,\\
			0^{l_1}1^{l_0}1^{l_0}1^{l_0}1^{l_0}0^{l_1} & \text{ otherwise.}
		\end{cases}
	$$
	We leave $g$ the same: $g=0^{{l_1 \over 2}-1}10^{{l_1\over 2}}0^{l_0}1^{l_0}1^{l_0}1^{l_0}0^{l_1}$.

	Instead of \emph{assignment gadgets}, we have \emph{vector gadgets}.
	$$\VG_1(\alpha_i)=Z_1 L V_0 R Z_2\text{ and }\VG_2(\beta_i)=V_1 D V_2,$$
	where $R=\bigcirc_{j \in [d]}\CG_1(\alpha_i,j), D=\bigcirc_{j \in [d]}\CG_2(\beta_i,j)$.

	Then, we replace the statement ``$\varphi$ is satisfied by $a_1 \vee a_2$'' with
	``vectors $\alpha_{i_1}$ and $\beta_{i_2}$ are orthogonal'' and the statement
	``$\varphi$ is satisfiable'' with ``there is a vector from the first set of variables
	and a vector from the second set of variables that are orthogonal''.

	For a vector $v$ and $k\in \{1,2\}$, we have $\VG_k'(v)=2^{T}\VG_k(v)2^{T}$, instead of $\AG_k'$.
	We set $f \in \{0,1\}^d$ to have $f[i]=1$ for all $i \in [d]$.

	We define the sequences as
	$$
		P_1=\bigcirc_{\alpha\in \seta}\VG_1'(\alpha),
	$$
	$$
		P_2=\left(\bigcirc_{i=1}^{n-1}\VG_2'(f)\right)
			\left(\bigcirc_{\beta\in \setb}\VG_2'(\beta)\right)
			\left(\bigcirc_{i=1}^{n-1}\VG_2'(f)\right).
	$$

	This completes the modification of the argument. We can check that we never use
	any property of $\cnfsat$ that $\Ort$ does not have.

\paragraph{Reducing $\MostOrt$ to Edit-Distance.}
Next, we modify the construction to show that Edit-Distance is a hard problem under a weaker
assumption, i.e., that the $\MostOrt$ problem does not have a truly sub-quadratic algorithm (Conjecture~\ref{conj:MOV}). 

\begin{theorem} \label{mostort_edit}
	Edit-Distance does not have strongly a subquadratic time algorithm
	unless $\MostOrt$ problem has a strongly subquadratic algorithm.
\end{theorem}
\begin{proof}
	We describe how to change the arguments from \cite{edit_hardness} to get the
	necessary reduction.
	We make all the modifications from the discussion above, as well as the following.

	We change $g$ as follows,
	$$	
		g=0^{{l_1 \over 2}-\left(1+{r \over d}2l_0\right)}1^{1+{r \over d}2l_0}
		0^{l_1 \over 2}
		0^{l_0}1^{l_0}1^{l_0}1^{l_0}0^{l_1}.
	$$

	We replace Lemma 1 from \cite{edit_hardness} with the following lemma.
	\begin{lemma}
		If $\alpha_{i_1}$ and $\beta_{i_2}$ are far vectors, then
		$$
			\EDIT(\VG_1(\alpha_{i_1}),\VG_2(\beta_{i_2}))\leq 2l_2+l+dl_0+k2l_0=:E_s.
		$$
	\end{lemma}
	\begin{proof}
		We do the same transformations of sequences as in Lemma 1 from \cite{edit_hardness} except that we
		get upper bound $E_s$ on the cost.
	\end{proof}

	We replace Lemma 2 from \cite{edit_hardness} with the following lemma.
	\begin{lemma}
		If $\alpha_{i_1}$ and $\beta_{i_2}$ are close vectors, then
		$$
			\EDIT(\VG_1(\alpha_{i_1}),\VG_2(\beta_{i_2}))= 2l_2+l+dl_0+k2l_0+d=:E_u.
		$$
	\end{lemma}
	\begin{proof}
		The proof proceeds along the same lines as the one for Lemma 2 from \cite{edit_hardness}.
	\end{proof}

	This finishes the description of the necessary changes.
\end{proof}

\subsection{Reducing $\MostOrt$ to $\DTWD$}

We are now ready to present our main reduction to DTWD.

\begin{theorem} \label{mostort_dtwd}
	If $\DTWD$ over sequences of symbols from an alphabet of size $5$ can be solved in strongly sub-quadratic time, then $\MostOrt$ can also be solved in truly sub-quadratic time.
\end{theorem}
\begin{proof}
The main arguments in this proof are provided in Lemmas~\ref{edit_leq_dtwd} and~\ref{dtwd_eq_edit} below. 
Here we explain why these two lemmas complete the proof of our theorem.

	Consider arbitrary sequences of symbols, $Q_1$ and $Q_2$.
	On the one hand, in Lemma \ref{edit_leq_dtwd} we will show that for a simple transformation $f$,
$$
		\EDIT(Q_1,Q_2)\leq \DTWD(f(Q_1),f(Q_2)).
$$
	
	On the other hand, in Lemma~\ref{dtwd_eq_edit} below we will show that
$$
		\EDIT(P_1',P_2')\geq \DTWD(f(P_1'),f(P_2')),
$$
	if $P_1'$ and $P_2'$ are the sequences constructed in Theorem \ref{mostort_edit}.

	Together, the two inequalities imply
	that $\EDIT(P_1',P_2')=\DTWD(f(P_1'),f(P_2'))$.
	This implies that we have the same hardness result for $\DTWD$ that we had for Edit-Distance, under the assumption that
	$f$ is a simple transformation. We will see that $f$ is indeed a very simple transformation, i.e.,
	$f(P_1')$ and $f(P_2')$ can be computed in time $O(|P_1'|)$ and $O(|P_2'|)$.

	$P_1'$ and $P_2'$ are sequences of symbols over an alphabet of size $4$. Transformation $f$
	introduce an extra symbol. Thus, the final sequences will be over an alphabet of size $5$.
\end{proof}

For an alphabet $\Sigma$, a symbol $a \not \in \Sigma$, a sequence
$Q=q_1q_2...q_p \in \Sigma^p$ of length $p$, and a vector $r$
of $p+1$ positive integers, we define the operation
$$
	A_a^r(Q):=a^{r_1}q_1a^{r_2}q_2a^{r_3}...a^{r_p}q_pa^{r_{p+1}}.
$$

\begin{lemma} \label{edit_leq_dtwd}
For any two sequences $Q_1 \in \Sigma^m$ and $Q_2 \in \Sigma^n$
of length $m$ and $n$, respectively,
$$
\EDIT(Q_1,Q_2)\leq \DTWD(A_a^{r_1}(Q_1),A_a^{r_2}(Q_2))
$$
holds for any two positive integer vectors $r_1$ and $r_2$.
\end{lemma}
\begin{proof}
In this proof, we will use use the following equivalent definition of Edit-Distance that will simplify the analysis.

\begin{observation} \cite{edit_hardness}.
\label{no_insertion}
	For any two sequences $x,y$, $\EDIT(x,y)$ is equal to the minimum, over all sequences $z$, of the number of deletions and substitutions needed to transform $x$ into $z$, and $y$ into $z$.
\end{observation}

Below we will write $A$ instead of $A_a^r$.

We will show how to convert a traversal of $A(Q_1)$ and $A(Q_2)$ achieving $\DTWD$ cost $\DTWD(A(Q_1),A(Q_2))$, into a transformation of
$Q_1$ and $Q_2$ into the same sequence. 
Using Observation~\ref{no_insertion}, we will conclude that the edit cost of the resulting transformations
will be at most $\DTWD(A(Q_1),A(Q_2))$, which is what we need to complete the proof.

Consider an optimal $\DTWD$ traversal of $A(Q_1)$ and $A(Q_2)$.
At any moment, we say that a marker in $A(Q_1)$ or in $A(Q_2)$ is
 of $\Sigma$ type iff the symbol it points to is in $\Sigma$, i.e., 
it is not equal to $a$.
We say that a symbol is of $\Sigma$ type iff it is in $\Sigma$.

From now on we consider only moments during the traversal of $A(Q_1)$ and $A(Q_2)$
when one or the other, or both markers change their type.
We can assume that, whenever both markers change their type, it is not the case
that before the change, the markers have different type. Indeed, if this happens,
we can replace the simultaneous change of type by two consecutive changes of type, and
this modification will not change the cost.
Consider any maximal contiguous subsequence of the sequence of moments during which only one of
the markers changes its type (the marker might change its type during the subsequence more than one time). 
We claim that any 
such contiguous subsequence of moments must have an even length.
Assume that this in not the case and
consider the earliest such subsequence that has an odd length.
Consider the type of the markers immediately before the last moment in the subsequence.
Because we considered the first subsequence with an odd length, and both sequences
start with symbols that are not of $\Sigma$ type, we get that immediately before the last moment, both
 markers must have the same type. WLOG, assume that the last change of type happens to the first marker and note that
immediately after the last change the markers have different type. At the next moment from the sequence,
either both markers change type (which, by our observation that before a simultaneous change of type both markers must of the same type, is impossible) or only the second marker changes its type.
Thus, we have found two consecutive moments from the sequence of moments in which the type changes, with the following three properties.
\begin{enumerate}
\item None of the two changes of type are simultaneous for both markers;
\item Both changes of type are not made by the same marker;
\item Before the first change of type, the markers have the same type.
\end{enumerate}
We count $\DTWD$ cost of any traversal as follows. Every jump (performed by
one of the markers or performed by both markers simultaneously),
contributes $1$ to the final cost of the traversal iff the symbols
that the markers point at immediately after the jump are different (contribution is
$0$ if the symbols are the same).
For two symbols $x$ and $y$, $1_{x \neq y}$ is equal to $1$ if $x \neq y$ and is equal to $0$ otherwise.
We set $x$ to be equal to the symbol that the marker that participates in the first change
of type points at \emph{after} the jump.
We set $y$ to be equal to the symbol that the marker that participates in the second change
of type points at \emph{after} the jump.

The first change of the type contributes $1$ to the final cost of $\DTWD(A(Q_1),A(Q_2))$
(we consider the corresponding jump to the change of the type and its contribution)
and the second change of the type contributes $1_{x \neq y}$ to the final cost. 
We can check that the two changes
can be replaced by a single simultaneous change in both sequences
by changing the traversal of $A(Q_1)$ and $A(Q_2)$ 
(the fact that we can to this follows from the definition of $A$). 
The simultaneous change costs $1_{x \neq y}$ and, therefore, we decrease the
cost of $\DTWD$ by $1$. This contradicts the assumption that
we consider an optimal traversal. Therefore, the assumption that
there exists a maximal contiguous subsequence of moments
during which only one of the markers changes type and the subsequence
is of odd length, is wrong.

Now we can partition the entire sequence of changes of type into
two kinds of contiguous subsequences that do not overlap.
\begin{enumerate}
\item A simultaneous change of type by both markers;
\item Two changes of type following one another made
by the same marker. None of the two changes are simultaneous.
\end{enumerate}

We will now show the promised conversion of the $\DTWD$ traversal of
$A(Q_1)$ and $A(Q_2)$ into an Edit-Distance transformation of $Q_1$ and $Q_2$
into the same sequence (as in Observation \ref{no_insertion}) such that the cost only decreases. This will finish the proof that
$\EDIT(Q_1,Q_2)\leq \DTWD(A(Q_1),A(Q_2))$.

We analyze both types of subsequences.
\begin{enumerate}
\item From the properties of the partition and the fact that
both $A(Q_1)$ and $A(Q_2)$ start with a symbol of $\Sigma$ type,
we get that before and after the change of type both markers are of
the same type.

{\bf Case 1.}
Both markers before the simultaneous change are of $\Sigma$ type.
Suppose that the markers point to symbols $x \in \Sigma$ and $y \in \Sigma$.
In this case we perform substitution of $x$ with $y$ when
transforming $Q_1$ and $Q_2$ into the same sequence.

{\bf Case 2.}
Both markers before the simultaneous change are not of $\Sigma$ type.
In this case we do not have a corresponding substitution or deletion
when transforming $Q_1$ and $Q_2$ into the same sequence.

We see that in both cases the performed actions before 
(contribution to $\DTWD(A(Q_1),A(Q_2))$) 
and after (contribution to $\EDIT(Q_1,Q_2)$)
the conversion cost the same.
\item
Similarly as in the previous kind of subsequence,
we conclude that before the first change of type, the
markers are of the same type.
We consider both possible cases.

{\bf Case 1.}
Both markers before the first change of type are of $\Sigma$ type.
Suppose that the markers point to symbols $x \in \Sigma$ and $y \in \Sigma$.
If $x \neq y$, we perform a substitution of $x$ with $y$ when
transforming $Q_1$ and $Q_2$ into the same sequence.
If $x=y$, we don't do anything.

{\bf Case 2.}
Both markers before the first change of type are not of $\Sigma$ type.
WLOG, the first marker changes the type twice. Before the
second change, the first marker points to $x \in \Sigma$.
We delete $x$ when performing the transformation of $Q_1$ and $Q_2$ into
the same sequence.

We can check that in the first case the cost
after the conversion can only be smaller than before
the conversion.
In the second case the costs before (contribution to $\DTWD$) and after
(contribution to Edit-Distance) the conversion are the same.
\end{enumerate}
\end{proof}

From now on, $\Sigma=\{0,1,2,3\}$ and $a=4$.

\begin{lemma} \label{dtwd_eq_edit}
	For some vectors $r_1$ and $r_2$ with positive, bounded
	integer coordinates,
	$$
		\EDIT(P_1',P_2')\geq \DTWD(A^{r_1}(P_1'),A^{r_2}(P_2')),
	$$
	where $P_1'$ and $P_2'$ are the sequences defined in Theorem \ref{mostort_edit}.
\end{lemma}
\begin{proof} We use notation from Theorem \ref{mostort_edit}.
	By $A'$ we will denote a transformation $A^r$ with $r_i=1$ for all $i$.

	Let $r_3$ be such that for all $k \in \{1,2\}$,
	$$
		A^{r_3}(\VG_k'(a))=A'(2^T)A'(\VG_k(a))A'(2^T).
	$$

	We set 
	$$
		A^{r_1}(P_1')=A'(3^{|P_2'|})A^{r_1'}(P_1)A'(3^{|P_2'|}),
	$$
	where $r_1'$ is such that
	$$
		A^{r_1'}(P_1)=\bigcirc_{a_1 \in A_1}A^{r_3}(\VG_1'(a_1)).
	$$
	We set
	$$
		A^{r_2}(P_2')=A^{r_2}(P_2)
	$$
	$$
		=\left(\bigcirc_{i=1}^{\nn-1}A^{r_3}(\VG_2'(f))\right)
                	\left(\bigcirc_{a_2\in A_2}A^{r_3}(\VG_2'(a_2))\right)
                	\left(\bigcirc_{i=1}^{\nn-1}A^{r_3}(\VG_2'(f))\right).
	$$

	We will use the following lemma to prove the inequality.
	\begin{lemma} \label{ag_eq}
		For vectors $\alpha,\beta \in \{0,1\}^d$,
		$$
			\EDIT(\VG_1(\alpha),\VG_2(\beta))
			\geq \DTWD(A'(\VG_1(\alpha)),A'(\VG_2(\beta))).
		$$
	\end{lemma}
	\begin{proof} 
		We consider two cases.

		{\bf Case 1.} The vectors $\alpha$ and $\beta$ are far. 
		In this case, we traverse the $A'(Z_1L)$
		part of $A'(\VG_1(\alpha))$ while the marker in $A'(\VG_2(\beta))$ stays
		at the first symbol. Then, we traverse the remaining part $A'(V_0RZ_2)$
		of $A'(\VG_1(\alpha))$ in parallel with $A'(\VG_2(\beta))$.
		We can check that we achieve $\DTWD$ cost equal to 
		$E_s=\EDIT(\VG_1(\alpha),\VG_2(\beta))$.

		{\bf Case 2.} The vectors $\alpha$ and $\beta$ are close.
		In this case, we traverse $A'(Z_1LV_0)$ and $A'(\VG_2(\beta))$ in parallel.
		Then, we traverse the $A'(RZ_2)$ part of $A'(\VG_1(\alpha))$ while
		the marker at $A'(\VG_2(\beta))$ stays at the last symbol.
		We can check that we achieve $\DTWD$ cost equal to 
		$E_u=\EDIT(\VG_1(\alpha),\VG_2(\beta))$.
	\end{proof}

	We are now ready to prove that 
	$$
		\EDIT(P_1',P_2')\geq\DTWD(A^{r_1}(P_1'),A^{r_2}(P_2')).
	$$
	We are going to show a $\DTWD$ traversal
	of $A^{r_1}(P_1')$ and $A^{r_2}(P_2')$ that achieves $\DTWD$ cost
	equal to $\EDIT(P_1',P_2')$. This will imply the inequality and
	will finish the proof.

	We proceed by considering two cases.

	{\bf Case 1.} There are two vectors $\alpha_{i_1}$ and $\beta_{i_2}$ from their respective
	sets that are far.
	We traverse $A'(\VG_1(\alpha_{i_1}))$ and $A'(\VG_2(\beta_{i_2}))$ as in Lemma \ref{ag_eq}
	achieving cost $E_s$.
	We traverse the rest of vector gadgets of $A^{r_1'}(P_1)$ with their
	counterparts from $A^{r_2}(P_2')$ as in Lemma \ref{ag_eq}. 
	When traversing the
	sequences $A'(2^T)$, we do that in parallel. When traversing $A'(2^T)$
	in parallel, it contributes nothing to the $\DTWD$ cost.

	We traverse the vector gadgets of $A^{r_2}(P_2')$ that are not traversed
	yet, as follows. We traverse the symbols that have $\Sigma$ type from $A^{r_2}(P_2')$
	with the $3$ symbols from $A^{r_1}(P_1')$ in parallel. We notice that we can do that
	in a way so that the $4$ symbols never contribute towards the final $\DTWD$ cost.
	Some of the $3$ symbols from $A^{r_1}(P_1')$ will still remain untraversed. We can
	traverse them while the second marker is on the last symbol of $A^{r_2}(P_2')$ 
	(it does not have $\Sigma$ type).

	By computing the cost of the traversal we get that it is 
	equal to $\EDIT(P_1',P_2')$.

	{\bf Case 2.} There is no pair of far vectors. This case is analogous to Case 1. The only difference
	is that we do not have two vectors $\alpha_{i_1}$ and $\beta_{i_2}$ to match.
	We choose them arbitrarily and then proceed as in the previous case.
	This finishes the analysis of this case.
\end{proof}

\section{Hardness for approximating DTWD}
\label{sec:approx}

In this section we prove that approximating DTWD and Frechet in certain settings, is ruled out by SETH.
The proofs involve simple modifications on the construction of Bringmann~\cite{Bring} for proving the hardness of Frechet, which are given in the following lemmas.


\begin{lemma} \label{frechet_not_metric}
	Given two lists $\{\alpha_i\}_{i \in [n]}$ and $\{\beta_i\}_{i\in [n]}$ of vectors
	$\alpha_i,\beta_i \in \{0,1\}^d$, we can construct two sequences $P_1$ and $P_2$
	in time $O(nd)$ such that $\F(P_1,P_2)=0$ if there are two vectors $\alpha_i$ and $\beta_j$
	that are orthogonal and $\F(P_1,P_2)=1$ otherwise. The sequences of points come from a pointset
	with constant number of points
	with a distance function $f$ between points that does not satisfy the triangle inequality.
\end{lemma}
\begin{proof}
	Let $Q_1=\{s_1,r_1,t_1,c_{1,0}^1,c_{1,1}^1,c_{1,0}^2,c_{1,1}^2\}$ be the poinset that 
	we will use for the construction of the first sequence.
	Let $Q_2=\{s_2,s_2^*,r_2,t_2,t_2^*,c_{2,0}^1,c_{2,1}^1,c_{2,0}^2,c_{2,1}^2\}$ be the pointset that 
	we will use for the construction of the second sequence.

	We set $f(v_1,v_2)=0$ for all
	$$(v_1,v_2) \in (\{s_2\}\times Q_1) \cup (\{s_1\}\times (Q_2 \setminus \{t_2^*\})) \cup \{(r_1,r_2)\} \cup(\{t_2\}\times Q_1) \cup(\{t_1\}\times Q_2\setminus\{s_2^*\})$$ 
	$$\cup \{(c_{1,0}^0,c_{2,0}^0),(c_{1,0}^0,c_{2,1}^0),(c_{1,1}^0,c_{2,0}^0)\}\cup \{(c_{1,0}^1,c_{2,0}^1),(c_{1,0}^1,c_{2,1}^1),(c_{1,1}^1,c_{2,0}^1)\}.$$
	We set $f(v_1,v_2)=1$ for all $v_1 \in Q_1$ and $v_2 \in Q_2$ that we did not set yet.
	$f$ is symmetric function, i.e., $f(v_1,v_2)=f(v_2,v_1)$ for all $v_1$ and $v_2$.

	We define \emph{coordinate gadget} for the first sequence as
	$$
		\CG_1(\alpha_i,j)=c_{1,(\alpha_i)_j}^{j\m{2}}
	$$ for $i \in [n]$ and $j \in [d]$.

	We define \emph{vector gadget} for the first sequence as
	$$
		\VG_1(\alpha_i)=r_1 \circ \bigcirc_{j \in [d]}\CG_1(\alpha_i,j)
	$$ 
	for $i \in [n]$.

	We define \emph{coordinate gadget} for the second sequence as
	$$
		\CG_2(\beta_i,j)=c_{2,(\beta_i)_j}^{j\m{2}}
	$$ for $i \in [n]$ and $j \in [d]$.

	We define \emph{vector gadget} for the second sequence as
	$$
		\VG_2(\beta_i)=r_2 \circ \bigcirc_{j \in [d]}\CG_2(\beta_i,j)
	$$ 
	for $i \in [n]$.

	We define the final sequences
	$$
		P_1=\bigcirc_{i \in [n]}(s_1 \circ \VG_1(\alpha_i) \circ t_1)
	$$
	and
	$$
		P_2=s_2 \circ s_2^* \circ (\bigcirc_{i \in [n]}\VG_2(\beta_i))
			\circ t_2^* \circ t_2.
	$$

	The rest of the proof follows from Claims \ref{orthog} and \ref{not_orthog}.
	
	\begin{claim} \label{orthog}
	If there are vectors $\alpha_i$ and $\beta_j$ that are orthogonal, then
	$\F(P_1,P_2)=0$.
\end{claim}
\begin{proof}
	We show traversal of $P_1$ and $P_2$ that achieves $\F(P_1,P_2)=0$.

	We stay at $s_2$ on $P_2$ and traverse the first sequence until
	we are at $s_1$ just before $\VG_1(\alpha_i)$. We stay at $s_1$
	on $P_1$ and traverse $P_2$ until we are at $r_2$ in $\VG_2(\beta_j)$.
	Now we perform simultaneous jumps in both sequences until we are
	at $\CG_1(\alpha_i,n)$ on $P_1$ and $\CG_2(\beta_j,n)$ on $P_2$.
	We perform jump to $t_1$ on $P_1$. Next we perform traversal of the
	rest of $P_2$ until we are at $t_2$ on $P_2$. Now we stay at $t_2$
	on $P_2$ and traverse the rest of sequence $P_1$ and we are done
	traversing both sequences. We can check that we achieve
	$\F(P_1,P_2)=0$.
\end{proof}

\begin{claim} \label{not_orthog}
	If there are no vectors $\alpha_i$ and $\beta_j$ that are orthogonal, then
	$\F(P_1,P_2)=1$.
\end{claim}
\begin{proof}
	We show that we can't traverse $P_1$ and $P_2$ in a way that achieves
	$\F(P_1,P_2)=0$. By the construction of $f$, we get that
	$\F(P_1,P_2)=1$.

	Suppose that this is not the case. Consider a traversal of $P_1$ and $P_2$
	that achieves $\F(P_1,P_2)=0$.
	
	There will be a moment during the traversal of $P_2$ when we are at $s_2^*$
	on $P_2$. At this very moment, by the construction of the distance function $f$, we must be
	at $s_1$ on $P_1$. Next jump is performed simultaneously on both sequences
	or on the second sequence only. In both cases we end up at $r_2$ on $P_2$.
	Let's denote this moment by $t$. We want to claim that there will be a moment
	when we are at $r_2$ on $P_2$ and at $r_1$ on $P_1$.
	If at this moment $t$ we are at $r_1$ on $P_1$, we have found the desired moment.
	Otherwise, consider the next moment $t'$ after $t$ when we are at $r_1$ on $P_1$.
	We claim that at this moment $t'$ we are at $r_2$ on $P_2$. Indeed, by the construction
	of $f$, we must be at $r_2$ or $t_2$ on $P_2$. We can't be at $t_2$ because
	we can only get there by traversing $t_2^*$ from $P_2$ but this requires
	being at $t_1$ on $P_1$. So we have found a moment when we are at $r_1$ on
	$P_1$ and at $r_2$ on $P_2$. By the construction of $f$ and the requirement
	that we achieve Frechet cost $0$, we conclude that we need to traverse
	the corresponding vector gadgets by doing simultaneous jumps. Given
	that there are no two vectors $\alpha_i$ and $\beta_j$ that are parallel
	and by the construction of $f$, we get that we can't achieve Frechet cost $0$.
\end{proof}

\end{proof}

An interesting corollary of Lemma~\ref{frechet_not_metric} is that any constant-factor approximation algorithm for the Frechet, cannot run in truly sub-quadratic time under SETH.
However, in the above reduction, the distance function over the points of the sequences violates the triangle-inequality.
Next, we show hardness for $(3-\eps)$-approximation algorithms when the points come from a metric.

\begin{lemma} \label{frechet_metric}
	Given two lists $\{\alpha_i\}_{i \in [n]}$ and $\{\beta_i\}_{i\in [n]}$ of vectors
	$\alpha_i,\beta_i \in \{0,1\}^d$, we can construct two sequences $P_1$ and $P_2$
	in time $O(nd)$ such that $\F(P_1,P_2)=1/2$ if there are two vectors $\alpha_i$ and $\beta_j$
	that are orthogonal and $\F(P_1,P_2)=3/2$ otherwise. Sequences of points come from a metric
	that has constant number of points.
\end{lemma}
\begin{proof}
	Let us denote the distance function that satisfies the triangle inequality by $f'$. Let $f$ be the distance
	function from Lemma \ref{frechet_not_metric}.
	We set $f'(v_1,v_2):=f(v_1,v_2)+1/2$ for $v_1 \in Q_1$ and $v_2 \in Q_2$.
	We set $f'(v_1,v_2)=1$ if $v_1,v_2 \in Q_1$ or $v_1,v_2 \in Q_2$.

	We can check that $f'$ satisfies the triangle inequality.
\end{proof}

Finally, we observe that this construction implies that if the distance function is arbitrary, then DTWD cannot be approximated to within any constant factor in truly sub-quadratic time, under SETH, as well.

\begin{lemma} \label{dtwd_not_metric}
	Given two lists $\{\alpha_i\}_{i \in [n]}$ and $\{\beta_i\}_{i\in [n]}$ of vectors
	$\alpha_i,\beta_i \in \{0,1\}^d$, we can construct two sequences $P_1$ and $P_2$
	in time $O(nd)$ such that $\DTWD(P_1,P_2)=0$ if there are two vectors $\alpha_i$ and $\beta_j$
	that are orthogonal and $\DTWD(P_1,P_2)=1$ otherwise. Sequences of points come from a pointset
	with constant number of points
	with a distance function $f$ between points that does not satisfy the triangle inequality.
\end{lemma}
\begin{proof}
	The sequences $P_1$ and $P_2$ are the same as in the proof of Lemma \ref{frechet_not_metric}.

	If there are two vectors $\alpha_i$ and $\beta_j$ that are orthogonal, we traverse the sequences
	in the same way as in Lemma \ref{orthog}.

	If there are no two vectors $\alpha_i$ and $\beta_j$ that are orthogonal, we show that
	$\DTWD(P_1,P_2)\geq 1$ in the same way as we do in Lemma \ref{not_orthog}. The following
	traversal achieves $\DTWD(P_1,P_2)=1$. We stay at $s_1$ on $P_1$ and traverse the entire $P_2$ (this costs $1$).
	We stand at $t_2$ and traverse the entire $P_1$ (this costs $0$).
\end{proof}

\section{Hardness for $k$-LCS}
\label{sec:klcs}

In this section we prove Theorem~\ref{thm:klcs}, along with another interesting lower bound for a variant of $k$-LCS (Theorem~\ref{thm:Local-k-LCS}).

As in the reduction to LCS, it will be much more convenient to reduce to the weighted version of the problem, defined below, as an intermediate step.

\begin{definition}[$k$-LCS and $k$-WLCS] An algorithm for $k$-LCS problem outputs the answer to the following question.
Given $k$ strings of length $n$ over alphabet $\Sigma$, what is the length of the longest sequence that appears in all $k$ strings as a subsequence?
In $k$-WLCS we are also given a scoring function $w:\Sigma \to [K]$ and the goal is to find the common subsequence $X$ of all $k$ strings that maximizes the sum $\sum_{i=1}^{|X|} w(X[i])$.
\end{definition}

As before, we can think of the common subsequence as a matching of the strings.
We can also adapt the previous proof to show a reduction from the weighted version to the unweighted version.

\begin{lemma}
\label{lem:kwlcs}
Computing the $k$-WLCS of $k$ strings of length $n$ over $\Sigma$ with weights $w: \Sigma \to [K]$ can be reduced to computing the $k$-LCS of $k$ strings of length $O(Kn)$ over $\Sigma$.
\end{lemma}
\begin{proof}
The proof is similar to the proof of Lemma~\ref{lem:wlcs} where we only had two strings, we will only outline the differences.
As before, the reduction maps each symbol $\ell$ into an interval of $w(\ell)$ copies of the same symbol $\ell$.
First, we can map a subsequence $X$ of the weighted instance of weight $w(X)$ into a subsequence of length $w(X)$ of the unweighted instance by mapping each symbol of $X$ into an interval.
Second, we can modify a subsequence of length $|X|$ of the unweighted instance into a subsequence of length at least $|X|$ which has the property that complete intervals are matched in the corresponding matching. 
Once we have this property we can contract each interval back into the original weighted symbol that generated it and obtain a subsequence of weight at lest $|X|$.
As before, these modifications can be done by scanning the strings from left to right and repeatedly converting each matching of parts of intervals into a matching of complete intervals while removing conflicting matches.
Each such modification adds $w(\ell)$ $k$-tuples to the matching and removes up to $w(\ell)$ previously matched $k$-tuples. 
The argument here is similar to the one in Lemma~\ref{lem:wlcs}, and is based on the observation that all conflicting $k$-tuples must come from the same interval in at least one of the $k$ strings.
\end{proof}

\subsection{$k$-Orthogonal-Vectors}

We will prove SETH-based lower bounds for problems on $k$ sequences 
via the orthogonal vectors problem on $k$ lists (see Lemma \ref{lem:kmaxsat} below).

\begin{definition}[$\kOrt$]
        Given $k$ lists $\{\alpha_i^t\}_{i \in [n]}$ ($t \in [k]$) of vectors $\alpha_i^t \in \{0,1\}^d$,
        are there $k$ vectors $\alpha_{i_1}^1,\alpha_{i_2}^2,...,\alpha_{i_k}^k$ that satisfy,
        $\sum_{h=1}^d \prod_{t \in [k]} \alpha_{i_t}^t[h] = 0$?
        Any collection of vectors $(\alpha_{i_t}^t)_{t \in [k]}$ with this property will be called orthogonal.
\end{definition}

\begin{definition}[$\kMostOrt$] \label{def:kmo}
        Given $k$ lists $\{\alpha_i^t\}_{i \in [n]}$ ($t \in [k]$) of vectors $\alpha_i^t \in \{0,1\}^d$
	and an integer $r \in \{1,2,...,d\}$,
        are there $k$ vectors $\alpha_{i_1}^1,\alpha_{i_2}^2,...,\alpha_{i_k}^k$ that satisfy,
        $\sum_{h=1}^d \prod_{t \in [k]} \alpha_{i_t}^t[h] \leq r$? 
        The LHS of the latter expression will be called the
	inner product of the $k$ vectors.
	A collection of vectors that satisfies the property will be called ($r$-)\emph{far}, and otherwise it will be called ($r$-)\emph{close}.
\end{definition}

\begin{lemma} \label{lem:kmaxsat}
	If $\kMostOrt$ on can be solved in $T(n,k,d)$ time, 
	then given a CNF formula on $n$ variables and $M$ clauses, 
	we can compute the maximum number of satisfiable clauses (MAX-CNF-SAT), 
	in $O(T(2^{n/k},k,M)\cdot \log{M})$ time.
\end{lemma}
\begin{proof} 
The proof is generalization of the one for Lemma \ref{lem:maxsat}.

Given a CNF formula on $n$ variables and $M$ clauses, split the variables into $k$ sets of size $n/k$ and list all $2^{n/k}$ partial assignments to each set.
Define a vector $v(\alpha)$ for each partial assignment $\alpha$ which contains a $0$ at coordinate $j \in [M]$ if $\alpha$ sets any of the literals of the $j^{th}$ clause of the formula to true, and $1$ otherwise. 
In other words, it contains a $0$ if the partial assignment satisfies the clause and $1$ otherwise.
Now, observe that if $\alpha_t$ (${t \in [k]}$) is assignment for variables of $t$-th set (every set if of size $n/k$), then the inner product of vectors $\{v(\alpha_t)\}_{t \in [k]}$ (as in definition \ref{def:kmo}) is equal to the number of clauses that the assignment $(\odot_{t \in k} \alpha_t)$ does not satisfy.
Therefore, to find the assignment that maximizes the number of satisfied clauses, it is enough to find $k$ vectors $\alpha_t$ (${t \in [k]}$) such that the inner product of vectors $\{v(\alpha_t)\}_{t \in [k]}$ is minimized.
The latter can be easily reduced to $O(\log{M})$ calls to an oracle for $\kMostOrt$ on $k$ sets of $N=2^{n/k}$ vectors each in $\{0,1\}^M$ with a standard binary search. 
\end{proof}



\subsection{Adapting the reduction}

There are two challenges in adapting the hardness proof for problem of computing $\LCS$ between two sequences to the problem of computing $\LCS$ between $k>2$ sequences: constructing the vector gadgets, and combining the gadgets in a way that implements a selection-gadget. 
We will start with the vector gadgets.

\paragraph{Vector gadgets.}
	We will need symbols $a,b,c,d$ with $w(a)=w(b)=w(c)=1$ and $w(d)=4^k$.
	For an integer $p\in \{0,1,2,...,2^k-1\}$ we define $v_p \in \{0,1\}^k$ to be a vector
	containing the binary expansion of $p$, i.e., $(v_p)_t$ is $t^{th}$ bit in the binary expansion of $p$,
	for $t \in [k]$. Let function $f$ satisfy $f(0)=a$ and $f(1)=b$. For $x \in \{0,1\}$, $\overline{x}:=1-x$.

	For $t$-th set of vectors $\{\alpha_i^t\}_{i \in [n]}$ ($t \in [k]$)
	and $i \in [n]$, and $j \in [d]$ we define \emph{coordinate gadget}
	$$
		\CG_t(\alpha_i^t,j)=
		\begin{cases} 
			dcd \bigcirc_{p=0}^{2^k-2}(f((v_p)_t)\circ d) & \text{if } (\alpha_i^t)_j=0 \\
			dd \bigcirc_{p=0}^{2^k-2}(f(\overline{(v_p)_t})\circ d) & \text{otherwise.}
		\end{cases}
	$$

	\begin{claim} \label{coordinate_gadget}
		Let $E_o^c=2+2^k\cdot w(d)$ and $E_n^c=E_o^c-1$. For $j \in [d]$
		and $i_1,i_2,...,i_k \in [n]$,
		$$
			\WLCS(\CG_1(\alpha_{i_1}^1,j),\CG_2(\alpha_{i_2}^2,j),...,\CG_k(\alpha_{i_k}^k,j))
			=\begin{cases}
				E_n^c & \text{if } (\alpha_{i_t}^t)_j=1 \text{ for all } t \in [k],\\
				E_o^c & \text{ otherwise.}
			\end{cases}
		$$
	\end{claim}
	\begin{proof} The main idea behind the construction of the coordinate gadgets is as follows.
		Fix $j \in [d]$ and consider a collection of $k$ vectors. Consider the $j^{th}$ coordinate
		of all the vectors. Let $c_1$, $c_2$, ..., $c_k$ be such that $c_t$ is equal to the $j^{th}$ coordinate of the $t^{th}$ vector.
		Suppose that for the $t^{th}$ sequence we set the coordinate gadget corresponding to $c_t$ to be
		equal to the following sequence. If $c_t=0$, we take binary expansion of the integers from $0$ to $2^k-1$
		and take $t^{th}$ bit from the expansion and concatenate all $2^k$ bits.
		If $c_t=1$, we do the same except we flip all the bits. Now consider the $\WLCS$ between all $k$
		sequences defined this way. For now, assume that we do not align symbols that have
		different indices, i.e., for two sequences $\alpha'$ and $\alpha''$, we are allowed to align
		$\alpha'[h']$ and $\alpha''[h'']$ iff $h'=h''$. (We take care of this assumption below.)
		We can easily see that the $\WLCS$ is \emph{always} 
		equal to $2$ between the sequences (independently of the values of $c_t$).
		Now let us modify the coordinate gadgets as follows. Instead of concatenating the bits corresponding
		to the integers from $0$ to $2^k-1$, we concatenate the bits for the integers from $0$ to $2^k-2$.
		We can check now that the $\WLCS$ is always equal to $2$ except when all the $c_t$ bits are
		equal (i.e., $c_t=0$ for all $t \in [k]$ or $c_t=1$ for all $t \in [k]$). If all the bits are equal,
		then the $\WLCS$ is equal to $1$. We want the construction of clause gadgets to satisfy the following property.
		If there exists $t \in [k]$ with $c_t=0$, then the $\WLCS$ is equal to some fixed large value. While, if $c_t=1$
		for all $t \in [k]$, then the $\WLCS$ should be equal to some fixed small value. Our current construction almost
		satisfies this property. We want to modify the construction so that the value of the $\WLCS$
		is equal to $2$ when $c_t=0$ for all $t \in [k]$. We can do that as follows. We take the previous
		construction and append a special symbol $c$ at the beginning of the binary sequence if 
		$c_t=0$. We can check that the construction satisfies the needed property under the stated assumption.
		We proceed by showing that the actual definition of clause gadgets removes the
		necessity of the assumption.

		We want to match all the $d$ symbols from every sequence, since if we don't do that
		we end up with a $\WLCS$ cost that is less than $E_o^c$. We proceed by
		assuming that we match all the $d$ symbols. We can now check that we
		have two matches if not all the vectors have a $1$ at the $j^{th}$ coordinate, while we have one match otherwise.
	\end{proof}

	Let $e$ be a symbol with $w(e)=100\cdot E_o^c$.

	For the $t$-th set of vectors $\{\alpha_i^t\}_{i \in [n]}$ ($t \in [k]$)
	and $i \in [n]$ we define the \emph{vector gadget}
	$$
		\VG_t'(\alpha_i^t)=e\circ \bigcirc_{j \in [d]}(\CG_t(\alpha_i^t,j)\circ e).
	$$
		Let $E_o=(d-r) \cdot E_o^c+r \cdot E_n^c$ and $E_n=E_o-1$.

	\begin{claim} \label{vector_gadget_prime}

		For $i_1,...,i_k \in [n]$,
		$$
			\WLCS(\VG_1'(\alpha_{i_1}^1),\VG_2'(\alpha_{i_2}^2),...,\VG_k'(\alpha_{i_k}^k))
			=\begin{cases}
				\geq E_o & \text{ if } \alpha_{i_1}^1,\alpha_{i_2}^2,...,\alpha_{i_k}^k \text{ are }r\text{-far,} \\
				\leq E_n &\text{otherwise.}
			\end{cases}
		$$
	\end{claim}
	\begin{proof}
		As in the proof of Claim \ref{coordinate_gadget}, we can conclude that in the
		optimal matching we use all the $e$ symbols from all the sequences. If this is not so, then the
		maximum $\WLCS$ score is $\leq E_n$.

		Using Claim \ref{coordinate_gadget} we can check that the $\WLCS$ cost is at least $E_o$, if
		the vectors $\alpha_{i_1}^1,\alpha_{i_2}^2,...,\alpha_{i_k}^k$ are $r$-far.
		Also, we can check that, if the vectors are $r$-close, then the $\WLCS$ cost is at most $E_n$.
	\end{proof}

	Let $f$ be a symbol with $w(f)=E_n$.
	For a vector $\alpha$ we define
	$$
		\VG_1(\alpha)=f \circ \VG_1'(\alpha),
	$$
	$$
		\VG_t(\alpha)=\VG_t'(\alpha) \circ f,
	$$
	for $t \in \{2,3,...,k\}$.

	\begin{claim} \label{vector_gadget}
		For $i_1,...,i_k \in [n]$,
		$$
			\WLCS(\VG_1(\alpha_{i_1}^1),\VG_2(\alpha_{i_2}^2),...,\VG_k(\alpha_{i_k}^k))
			=\begin{cases}
				\geq E_o & \text{ if } \alpha_{i_1}^1,\alpha_{i_2}^2,...,\alpha_{i_k}^k \text{ are }r\text{-far,} \\
				E_n &\text{otherwise.}
			\end{cases}
		$$
	\end{claim}
	\begin{proof}
		If the vectors $\alpha_{i_1}^1,\alpha_{i_2}^2,...,\alpha_{i_k}^k$ are $r$-far,
		we have a $\WLCS$ cost of at least $E_o$ as in Claim \ref{vector_gadget_prime} and we do not use
		any of the $f$ symbols. We cannot achieve a larger score than $E_0$ by using the $f$ symbols.
		
		If the vectors are $r$-close and we do not use any $f$ symbols, the maximum cost is at most $E_n$ by Lemma \ref{vector_gadget_prime}.
		If it is less than that, we can use the $f$ symbols and achieve a score of $E_n$. Notice that, if we use the $f$ symbols,
		we cannot use any other symbol in any matching.		
	\end{proof}

\paragraph{Combining the vector gadgets.}
A very simple padding strategy implies the lower bound for a variant of $k$-LCS.

\begin{definition}[Local-$k$-LCS]
Given $k$ strings of length $n$ over an alphabet $\Sigma$ and an integer $L$, what is the length of longest sequence $X$ such that there are $k$ substrings of length $L$, one of each input string, such that $X$ is a common subsequence of each one of these substrings.
\end{definition}

In words, we are looking for substrings of length $L$ for which the LCS score is maximized.

\begin{theorem}
\label{thm:Local-k-LCS}
If Local-$k$-LCS on strings of length $n$ over an alphabet of size $O(1)$ can be solved in $O(n^{k-\eps})$ time, for some $\eps>0$, then SETH is false.
\end{theorem}

Theorem~\ref{thm:Local-k-LCS} follows from the following reduction.
 We note that in the constructed instances, $L$ is always polylogarithmic in the lengths of the sequences, and therefore the problem can easily be solved in $\tilde{O}(n^k)$ time.
This problem is closely related to the \emph{Normalized-LCS} problem which was studied in \cite{AP01,EL04} and for which an $n^{2-o(1)}$ lower bound based on SETH was shown in \cite{AVW14}.

\begin{lemma}
$k$-Most-Orthogonal Vectors on $k$ lists of $N$ vectors in $\{0,1\}^M$ can be reduced to Local-$k$-LCS on $k$ strings of length $2^k \cdot N \cdot M^{O(1)}$ over an alphabet of size $O(1)$.
\end{lemma}

\begin{proof}
We construct $k$ lists of vector gadgets from our $k$ lists of vectors as in the above discussion.
By the reduction of Lemma~\ref{lem:kwlcs} from WLCS to LCS, we can convert each vector gadget $VG_t(\alpha^t)$ to a longer string $UVG_t(\alpha^t)$ such that what we proved for WLCS in Claim~\ref{vector_gadget} holds for LCS instead.
Let $L$ be the length of the longest vector gadget $UVG_t(\alpha^t)$ that we create in this process.
We also introduce two new symbols $x,y$.
The first string will be defined as $P_1 = \bigcirc_{i=1}^N (UVG_1(\alpha^1_i) \circ x^L )$, while the other $k-1$ strings will be $P_t = \bigcirc_{i=1}^N (UVG_t(\alpha^t_i) \circ y^L )$, for $t = 2$ to $k$.
To complete the reduction, we claim that if the input is a YES instance of $k$-Most-Orthogonal Vectors, there will be $k$ substrings of length $L$ with LCS $\geq E_o$, namely the $k$ vector gadgets corresponding to the $r$-far vectors, while otherwise the maximum score of any $k$ substrings is $E_n$. The latter part is implied by Claim~\ref{vector_gadget} and by noting that the $x,y$ parts can never be matched, and they are long enough to prevent any substring of length $L$ to contain symbols from more than one vector gadget.

\end{proof}

Next, we focus on the classic $k$-LCS problem and show how to implement the selection-gadget while making the existence of orthogonal vector influence the LCS in a manageable way.
Unfortunately, we are not able to do this without introducing $O(k)$ new symbols to the alphabet.

Our lower bound for $k$-LCS (Theorem~\ref{thm:klcs}) follows from the following reduction.

\begin{lemma}
For any $k \geq 2$, $k$-Most-Orthogonal Vectors on $k$ lists of $n$ vectors in $\{0,1\}^d$ can be reduced to $k$-LCS on $k$ strings of length $k^O(k) \cdot n \cdot d^{O(1)}$ over an alphabet of size $O(k)$.
\end{lemma}

\begin{proof}
We will show a reduction to $k$-WLCS and use Lemma~\ref{lem:kwlcs} to conclude the proof.

We construct $k$ lists of vector gadgets from our $k$ lists of vectors as in the above discussion.
Let $D$ be the maximum possible sum of weights of all symbols in any vector gadget, and note that $D=\poly(2^k,d)$ and that $D>E_o$.
For $i \in \{2,\ldots,k\}$ we will introduce a new symbol $3_i$ to the alphabet, and set $B_k=B=(10kD)^2$ and for $2\leq i\leq k$ set $w(3_i) = B_i = 2k \cdot B_{i+1}$.
Finally, add two new symbols $0,2$ and set $w(0)=w(2)=C=10k^2 B_2$.
The weights achieve $C >> B_2 >> \cdots >> B_k = B >> D >> E_o$.

Our $k$ strings are defined as follows. For $i \in [k]$,

\[
P_i = (3_{i+1} \cdots 3_{k})^Q \circ (3_2 \cdots 3_{i}) \circ \left( VG'_i(f) \right)^{(i-1)N} \circ  \bigcirc_{t=1}^{N} VG'_i(\alpha^i_t) \circ \left( VG'_i(f) \right)^{(i-1)N} \circ (3_{i+1} \cdots 3_{k})^Q
\]
where $VG'_1(x) = 0 \circ VG_1(x) \circ 2$, $VG'_i(x) = 0 \circ VG_i(x) \circ 2 \circ (3_2 \cdots 3_{i})$ if $i \geq 2$, and $Q=|P_k|$. 

The intuition behind this padding is that we want to force any optimal matching to match all $n$ vector gadgets of the first string to precisely $n$ vector gadgets from each other string. 
This is achieved since: if at least one vector gadget from $P_i$ is not matched, we will lose some $0$ or $2$ symbols that we could have matched, while if more than $n$ vector gadgets are matched, we will lose at least one $3_i$ symbol.
In addition, as long as we match consecutive $n$ intervals from each string, we will get the same score from the padding, and therefore the optimal matching will be determined by the existence of an $r$-far set of vectors.
The WLCS will be $E$ if there are no $r$-far vectors, and $E+1$ if there are, for an appropriately defined $E$.

To make this argument more formal, we can follow the steps in the proof of Lemma~\ref{lem:LCSmain} for LCS of two strings. 
First, we can prove an analog of Claim~\ref{cl:1}, stating that matching $n'$ intervals (vector gadgets) in some $P_t$ for some $n' > n$ can only contribute up to $(n'-n)(B-1)$ to the score.
Then, we observe that by the padding construction, if $n' > n$ then we will not be able to match at least $(n'-n)$ of the $3_t$ symbols that we could have matched if $n'$ was equal to $n$, which incurs a loss much greater than $(n'-n)B$.
Therefore, in an optimal matching, exactly $n$ intervals will be matched in each sequence, and it is easy to see that the score is then determined by the existence of an $r$-far set of vectors.

Let $E_U=2C+E_n$ and $E_G= n\cdot E_U + B_2 + (2n+1)\cdot \sum_{i=2}^k B_i$.
The following two lemmas prove that there is a gap in the $\WLCS$ of our 
$k$ sequences when there is a collection of $k$ vectors that are $r$-far as 
opposed to when there is none.

\begin{lemma}
\label{lem:good}
If there is a collection of $k$ vectors that are far, 
then $\WLCS(P_1,\ldots,P_k) \geq E_G + 1$.
\end{lemma}

\begin{proof}
Let $t_1,\ldots,t_k$ be such that the $k$ vectors $(\alpha^i_{t_i})_{i=1}^k$ are $r$-far.

First, match the corresponding gadgets, $(VG_i(\alpha^i_{t_i}))_{i=1}^k$, along with the $0$ and $2$ symbols surrounding each of these gadgets, to get a weight of 
at least $2C+E_o=2C+E_n+1 = E_U+1$, by Claim~\ref{vector_gadget}.

Then, Match the $i_1-1$ vector gadgets (and the surrounding $0,2$ symbols) to the left of $VG_1'(\alpha^1_{t_1})$ to the $i_1-1$ vector gadgets immediately to the left of $VG_i'(\alpha^i_{t_i})$, for every $i \in \{2,\ldots,k\}$, and similarly, match the $n-i_1$ gadgets to the right. 
The total additional weight we get is at least $(n-1)\cdot E_U$.

Then, note that after the above matches, only $(n-1)$ out of the $(3n+1)$ $3_2$-symbols in $P_2$ are surrounded by matched symbols.
The remaining $(2n+2)$ $3_2$-symbols can be matched,  
giving an additional weight of $(2n+2)\cdot B_2$, as follows:
Consider the leftmost matched $0$ in $P_2$, call it $x$, and assume there are $m$ $3_2$-symbols to the left of it in $P_2$. Match these $3_2$-symbols to the $m$ such symbols in each other string $P_i$ that appear immediately to the left of the symbol that is matched our $x$.
By construction, and the fact that $m$ can be at most $n$, we know that there are enough matchable $3_2$ symbols in the other strings.

Then, similarly, note that at this point, only $3n$ out of the $(5n+1)$ $3_3$-symbols in $P_3$ are surrounded by matched symbols.
The remaining $(2n+1)$ $3_3$-symbols can be matched, as above,
for an additional weight of $(2n+1)\cdot B_3$.
And in general, we perform this process for $i$ from $2$ to $k$, and at $i^{th}$ stage, only $(2(i-2)n+n+1-1)$ out of the $(2(i-1)n+n+1)$ $3_i$-symbols in $P_i$ are surrounded by matched symbols, and we can match the remaining ones to get an additional weight of $(2n+1)\cdot B_i$.
Thus, the total contribution of the $3_i$ symbols is $B_2 + \sum_{i=2}^k (2n+1)B_i$.

The total weight of our matching is at least $E_U +1 + (n-1)\cdot E_U + B_2+ (2n+1)\cdot \sum_{i=2}^k B_i = E_G +1$.
\end{proof}

The hard part is upper bounding the score when there is no collection of $r$-far vectors, and we will spend the rest of the proof towards this end.

\begin{lemma}
If there is no collection of $k$ vectors that are far, 
then $\WLCS(P_1,\ldots,P_k) \leq E_G$.
\end{lemma} 

\begin{proof}
Consider any optimal matching of our $k$ strings. The goal is to bound its score by $E_G$.
Our plan will be to divide the contribution to the score into two: (a) the contribution of the vector gadgets, and (b) the contribution from the padding, i.e. the $3_i$ symbols.
In any matching, there is a tradeoff between the scores from (a) and (b): the more vector gadgets we align, the fewer $3_i$'s we can match, and vice versa.
We will prove upper bounds for both contributions and show that they imply an upper bound of $E_{G}$ on the total score.

We start by formally defining (a) and upper bounding it.

For each string $P_i$, let $s_i$ and $t_i$ be the first $0$ symbol and the last $2$ symbol from $P_i$ that are matched in our optimal matching, if they exist, respectively.
A simple observation is that if some $0$ symbol is matched in the optimal matching ($s_i$ exists for all $i \in [k]$), then there must exist some $2$ symbol that is also matched: otherwise, match the $2$ immediately following that $0$ and note that any conflicting matches must come from inside the vector gadgets and therefore removing all of them will decrease the score by much less than $w(2)$.
Thus, we can define $N_i$ to be the number of vector gadgets that lie between $s_i$ and $t_i$, and if such $s_i,t_i$ do not exist, we set $N_i=0$.
By construction, $N_i \leq 2(i-1)n+n$, for all $i \in[k]$.
Note that $(s_1,\ldots,s_k)$ and $(t_1,\ldots,t_k)$ must be in our matching. 

We will assume that $N_i \geq 1$ for all $i$, since the only other case is that $\forall i\in[k]: N_i=0$, which can easily be seen to be sub-optimal: in this case, only $3_i$ symbols are matched, and there cannot be more than $(2(i-1)n+n+1)$ matched $3_i$ symbols for any $i \in \{2,\ldots,k\}$ which implies the following upper bound on the score:
$\sum_{i=2}^k (2(i-1)n+n+1) B_i \leq 3kn\sum_{i=2}^k B_i \leq 3kn B_2 < n\cdot C < E_G$.

By construction, there are no $3_i$ symbols between $s_1$ and $t_1$, which implies that the matching in between $(s_1,\ldots,s_k)$ and $(t_1,\ldots,t_k)$ does not contain any $3_i$ symbols.
The total contribution of this part is what we call (a) above.
On the other hand, the matching to the left of $(s_1,\ldots,s_k)$ and to the right of $(t_1,\ldots,t_k)$ cannot contain anything besides $3_i$ symbols: If some symbol $\sigma \notin \{ 0,3_2,\ldots,3_k \}$ appears in $P_i$ before $s_i$ and is matched, then the $0$'s that appear right before the matched $\sigma$'s could have been matched together without any conflicts, which contradicts the optimality of the matching. 
An analogous argument shows that $t_i$ is to the right of any matched $\sigma \notin \{2,3_2,\ldots,3_k\}$.
Thus, the contribution of part (b) only comes from $3_i$ symbols.

This motivates the following definitions.
From now on, we will refer to the sequences composed of the vector gadgets that are surrounded by $0,2$ as ``intervals", i.e. sequences of the form $0 \circ VG_i(x) \circ 2$.
Consider the substrings between $s_i$ and $t_i$ in each string $P_i$ and remove any $3_i$ symbols in them - since they are not matched anyway - and note that we obtain a concatenation of $N_i$ intervals.
Moreover, by our assumption that there is no satisfying assignment, we know that for any choice of one interval from each string, the $k$-LCS is upper bounded by $E_U = 2C + E_n$, by Claim~\ref{vector_gadget}.
The main quantity we will be interested in is $W(L_1,\ldots,L_k)$ which is defined to be the maximum score of a matching of any $k$ strings $T_1,\ldots,T_k$ such that $T_i$ is the concatenation of $L_i$ intervals, and for any choice of one interval from each $T_i$, the optimal score is $E_U$.
By the symmetry of $k$-LCS, we can assume WLOG that $L_1\leq \cdots \leq L_k$, and otherwise we reorder.
To get the desired upper bound on $W(L_1,\ldots,L_k)$ it will be convenient to first upper bound $W_0(L_1,\ldots,L_k)$, which is defined in a similar way, except that we require the matching to match all $0$ and $2$ symbols from $T_1$, i.e. the string string with fewest intervals.

Define $E_B=2C + D$ which is an upper bound on the maximum possible total weight of all the symbols in an interval.
A key inequality, which we will use multiple times in the proof, following from the fact that the  $0$/$2$ symbols are much more important than the rest, is the following.
 
\begin{fact}
\label{klcs:ineq2} Our parameters satisfy $E_B < E_U + (B - 1)/(k-1)$.
\end{fact}
 \begin{proof} 
 Follows since $(k-1)(E_B-E_U) < (k-1)D < B$, by our choice of parameters.
 \end{proof}

\begin{claim}
\label{cl:k0}
For any integers $1 \leq L_1\leq\ldots\leq L_k $, we can upper bound $W_0(L_1,\ldots,L_k) \leq L_1 \cdot E_U + (L_k-L_1)\cdot (B-1)$.
\end{claim}

\begin{proof}
Let $T_1,\ldots,T_k$ be any $k$ sequences with $L_1,\ldots,L_k$ intervals, respectively, that satisfy the assumption in the definition of $W_0$. 
Consider an optimal matching of the $k$ sequences in which all the $0$ and $2$ symbols of $T_1$ are matched and we will upper bound its weight $E_F$ by $L_1 \cdot E_U + (L_k-L_1)\cdot (B-1)$, which will prove the claim.
Note that in such a matching, for any  $i \in \{2,\ldots,k\}$, each interval of $T_1$ must be matched completely within one or more intervals of $T_i$, and each interval of $T_i$ has matches to at most one interval from $T$ (otherwise, it must be the case that some $0$ or $2$ symbol in $T_1$ is not matched).

We upper bound the weight of the matching by considering two kinds of intervals in $T_1$ and upper bounding their contributions. 
Let $x$ be the number of intervals of $T_1$ that contribute at most $E_U$ to the weight of our optimal matching, and call the other $(L_1-x)$ intervals ``full".
Note that any full interval must be matched to a substring of $T_i$, for some $i \in \{ 2,\ldots,k\}$, that contains at least two intervals for the following reason. The $0$ and $2$ symbols of the interval of $T_1$ must be matched, and, if the matching stays within a single interval of $T_i$, for all $i \in \{ 2,\ldots,k\}$, and has more than $E_U$ weight, then we have a contradiction to the assumption that no $k$ intervals, one from each string, can have a $k$-LCS score greater than $E_U$.
Thus, we have $x$ intervals consuming at least $1$ interval from every $T_i$, and we have $(L_1-x)$ full intervals consuming at least $1$ interval from every $T_i$ and at least $2$ intervals from some $T_i$.
Using the fact that the total number of intervals in $T_2,\ldots,T_k$ is $L_2+\cdots+L_k \leq (k-1) L_k$, we get the condition, 
$$
(k-1) \cdot x+k \cdot(L_1-x) \leq (k-1) L_k .
$$
We can now give an upper bound on the weight of our matching, by summing the contributions of each interval of $T_1$:
There are $x$ intervals contributing $\leq E_U$ weight, and there are $(L_1-x)$ intervals with unbounded contribution, but we know that even if all the symbols of an interval are matched, it can contribute at most $E_B$.
Therefore, the total weight of the matching can be upper bounded by 
\[
E_F \leq (L_1-x) \cdot E_B + x \cdot E_U 
\]
We claim that no matter what $x$ is, as long as the above condition holds, this expression is less than $L_1 \cdot E_U + (L_k-L_1) \cdot (B-1)$.

To maximize this expression, we choose the smallest possible $x$ that satisfies the above condition, since $E_B > E_U$, which implies that $x = \max\{ 0, k L_1 - (k-1) L_k \}$.

First, consider the case where $L_k\geq L_1 \cdot \frac{k}{k-1}$, and therefore $x=0$,
which means that all the intervals of $T_1$ might be fully matched.
Using Fact~\ref{klcs:ineq2} and that $L_k-L_1 \geq L_1 / (k-1)$, we get the desired upper bound:
$$
E_F \leq L_1 \cdot E_B \leq L_1 \cdot (E_U + (B-1)/(k-1)) \leq L_1 \cdot E_U + (L_k-L_1) \cdot (B-1).
$$

Now, assume that $L_k < L_1 \cdot \frac{k}{k-1}$, and therefore $x = kL_1-(k-1)L_k$.
In this case, when setting $x$ as small as possible, the upper bound becomes:
\[
E_F 
\leq ( (k-1)L_k - (k-1)L_1 )\cdot E_B + (kL_1-(k-1)L_k) \cdot E_U 
= L_1 \cdot E_U + (k-1)(L_k-L_1) \cdot (E_B - E_U),
\]
which, by Fact~\ref{klcs:ineq2}, is less than $L_1 \cdot E_U + (L_k-L_1) \cdot (B-1)$.
\end{proof}

We are now ready to upper bound the more general $W(L_1,\ldots,L_k)$.

\begin{claim}
\label{cl:k1}
For any integers $1 \leq L_1\leq\ldots\leq L_k $, we can upper bound $W(L_1,\ldots,L_k) \leq L_1 \cdot E_U + (L_k-L_1)\cdot (B-1)$.
\end{claim}

\begin{proof}
We will prove by induction on $\ell \geq k$ that: for all $1 \leq L_1\leq\ldots\leq L_k$ such that $L_1+\cdots+L_k \leq \ell$, $W(L_1,\ldots,L_k) \leq L_1 \cdot E_U + (L_k-L_1) \cdot (B-1)$.

The base case is when $\ell=k$ and $L_1=\cdots = L_k=1$. Then $W(1,\ldots,1) = E_U$, by the assumption on the strings in the definition of $W$, and we are done.

For the inductive step, assume that the statement is true for all $\ell' \leq \ell-1$ and we will prove it for $\ell$.
Let $L_1,\ldots,L_k$ be so that $1\leq L_1\leq \cdots \leq L_k $ and $L_1+\cdots+L_k=\ell$ and let $T_1,\ldots,T_k$ be sequences with a corresponding number of intervals.
Consider the optimal (unrestricted) matching of $T_1,\ldots,T_k$, denote its weight by $E_F$. 
Our goal is to show that $E_F \leq L_1 \cdot E_U + (L_k-L_1) \cdot (B-1)$.

If every $0$/$2$ symbol in $T_1$ is matched, then, by definition, the weight cannot be more than $W_0(L_1,\ldots,L_k)$, and by Claim~\ref{cl:k0} we are done.
Otherwise, consider the first unmatched $0$/$2$ symbol in $T_1$, call it $x$, and there are two cases.

\text{ \bf The $x= 0$ case:} 
If $x$ is the first $0$ in $T_1$, then for some $i \in \{2,\ldots,k\}$, the first $0$ in $T_i$ must be matched to some $0$ after $x$ (otherwise we can a $0$ to the matching without violating any other matches) which implies that none of the symbols in the interval starting at $x$ can be matched, since such matches would be in conflict with the match that contains this first $0$.
Otherwise, consider the $2$ that appears right before $x$, call it $y$, and note that it must be matched, 
to some $2$-symbols $y_i$ in $T_i$ for every $i \in \{2,\ldots,k\}$,
 by our choice of $x$ as the first unmatched $0$/$2$ symbol in $T_1$. 
Now, there are two possibilities: 
either for some $i \in \{2,\ldots,k\}$, our $y_i$ is the very last $2$ in $T_i$, and there are no more intervals in $T_i$ after this match, 
or for some $i \in \{2,\ldots,k\}$, the $0$ right after $y_i$ is already matched to some $0$ in $T_1$ that is after $x$ (from a later interval in $T_1$).
Note that in either case, the interval starting at $x$ (and ending at the $2$ after it) is completely unmatched in our matching.

Let ${T_1'}$ be the sequence with $(L_1-1)$ intervals which is obtained from $T_1$ by removing the interval starting at $x$. 
The weight of our matching will not change if we look at it as a matching between $T_2,\ldots,T_k$ and ${T_1'}$ instead of $T_1$, which implies that $E_F \leq W(L-1_1,L_2,\ldots,L_k)$.
Using our inductive hypothesis we conclude that $E_F \leq (L_1-1) \cdot E_U + (L_k-L_1+1) \cdot (B-1) \leq L_1 \cdot E_U + (L_k-L_1) \cdot (B-1)$, since $E_U > B$, and we are done.

\text{ \bf The $x=2$ case:} The $0$ at the start of $x$'s interval must have been matched to some $0$-symbols $x_i$ from each string $T_i$. 
For each $i \in \{2,\ldots, k\}$, let $z_i$ be the $2$ at the end of $x_i$'s interval.
Note that for at least one $i \in \{2,\ldots, k\}$, $z_i$ must be matched to some $w=2$ in $T_1$ after $x$, since otherwise, we can add $(x,z_2,\ldots,z_k)$ to the matching, gaining a cost of $C$, and the only possible conflicts we would create will be with matches containing symbols inside the $x_i \to z_i$ interval (that are not $0$ or $2$), for some $i \in \{2,\ldots, k\}$, or inside $x$'s interval, and if we remove all such matches, we would lose weight of at most $(E_B-2C)$ which is much smaller than the gain of $C$ from the new $2$ we matched - implying that our matching could not have been optimal.
Let $j \in \{ 2,\ldots,k\}$ be the index of this string, so that in $T_j$, both $x_j$ and $z_j$ are matched. 
Therefore, there are $c \geq 2$ intervals in $T_1$ that are matched to a single interval in $T_j$: 
all the intervals starting at the $0$ right before $x$ and ending at $w$ are matched to the $x_j \to z_j$ interval.
Let $T_1'$ be the sequence obtained from $T_1$ by removing all these $c$ intervals and let $T_j'$ be the sequence obtained from $T_j$ by removing the $x_j \to z_j$ interval.
Similarly, define $T_i'$ for every $i \in [k]-\{1,j\}$ to be the sequence obtained from $T_i$ by removing all the $c_i \geq 1$ intervals starting at $x_i$ and ending at the $2$ that is matched with $z_j$.
Our matching can be split into two parts: a matching of $T_1',\ldots,T_k'$, and the matching of the $x_j \to z_j$ interval to the removed intervals.
The contribution of the latter part to the weight of the matching can be at most the weight of all the symbols in an interval, which is $E_B$.
Consider the new sequences $T_1',\ldots,T_k'$ and note that: for each $i$, $T_i$ contains no more than $L_i-1$ intervals while the sequence with fewest intervals has no more than $L_1-c$ which is the number of intervals in $T_1'$.
Thus, by definition, we know that any matching of $T_1',\ldots,T_k'$ can have weight at most $W(L_1-c,\ldots,L_k-1)$, and
by the inductive hypothesis, we can upper bound $W(L_1-c,\ldots,L_k-1)\leq (L_1-c) \cdot E_U + (L_k-1-L_1+c) \cdot (B-1)$.
Summing up the two bounds on the contributions, we get that the total weight of the matching is at most:
\[
E_F \leq E_B + (L_1-c) \cdot E_U + (L_k-L_1+c-1) \cdot (B-1)
\leq L_1 \cdot E_U + (L_k-L_1) \cdot (B-1) + (c-1) \cdot (B-1) + E_B - c \cdot E_U 
\]
However, note that $E_B < 1.1 E_U$ and that $(c-1.1) E_U > 10(c-1.1) B > (c-1)B$, which implies that
 $E_F$ can be upper bounded by $L_1 \cdot E_U + (L_k-L_1) \cdot (B-1)$, and we are done.
\end{proof}

We now turn to bounding (b).
Recall the definition of $N_i$ above, as the number of intervals from $P_i$ that are matched.
Let us also define $x_{i-}$ as the number of $3_i$ symbols from $P_i$ that appear before $s_i$ and are matched in our optimal matching, and define $x_{i+}$ to be the number of such $3_i$ symbols that appear after $t_i$.
Then, the contribution of (b) to the score can be bounded by $\sum_{i=2}^k (x_{i-}+x_{i+})B_i$.
A simple but key observation is the following.

\begin{claim}
\label{cl:tricky}
For every $i \in \{2,\ldots,k\}$, 
$$
x_{i-}+x_{i+} \leq 2(i-1)n+n+2 - \sum_{j=2}^{i-1} (x_{j-}+x_{j+}-1) - N_i
$$
\end{claim}
\begin{proof}
Focus on $P_i$ and note that there are only $(2(i-1)n+n+1)$ $3_i$-symbols in it.
To make the counting easier, let us define a set $U$ that is initially empty, and we will add unmatchable $3_i$ symbols, from $P_i$, to $U$. In the end, we will argue that $|U|+x_{i-}+x_{i+}$ must be at most $(2(i-1)n+n+1)$.

First, we add the $(N_i-1)$ $3_i$ symbols that lie between $s_i$ and $t_i$ to $U$, since those are clearly unmatchable.

Second, we will focus on the prefix of $P_i$ that ends at $s_i$, call it $Q_i$.
For $2 \leq j<i$, note that there must be $x_{j-}$ $3_j$-symbols in $Q_i$ that are matched and let $q_j$ be the first such $3_j$ symbol.
Since $q_j$ is matched to the first $3_j$ symbol in $P_j$ that is matched, and that in $P_j$ there are no $3_h$ symbols, for any $h>j$ between that $3_j$ symbol and $s_j$, we can conclude that: for any $j<h<i$, all the $x_{h-}$ $3_h$-symbols in $Q_i$ that are matched are in the subsequence of $Q_i$ starting at $q_h$ and ending at $q_j$.
In fact, this implies that all the $x_{h-}$ $3_h$-symbols in $Q_i$ that are matched are in the subsequence of $Q_i$ starting at $q_h$ and ending right before $q_{h-1}$.
Thus, for each $2 \leq h<i$, we can add $x_{h-}$ new $3_i$ symbols to our unmatchable $U$ - the ones in the latter subsequence.

Finally, we focus on the suffix of $P_i$ that starts at $t_i$, and using a similar reasoning we conclude that for each $2 \leq h<i$, we can add $(x_{h+}-1)$ new $3_i$ symbols to our unmatchable $U$.

Thus, we conclude that $(N_i-1)+\sum_{j=2}^{i-1} (x_{j-}+x_{j+}-1)+x_{i-}+x_{i+} \leq (2(i-1)n+n+1)$, which proves the claim.

\end{proof}

For any fixed values for $N_1,\ldots,N_k$ satisfying $1 \leq N_i \leq 2(i-1)n+n$, we can compute the largest possible contribution of part (b).
Since if $i<j$ then $B_i$ is much larger than $B_j$, the optimal score is achieved when setting $(x_{i-}+x_{i+})$ to be as large as possible, regardless of the $3_j$ symbols we make unmatchable for $j>i$.
That is, we claim that the optimal score is achieved when each of the inequalities in Claim~\ref{cl:tricky} are saturated, i.e. $x_{i-}+x_{i+} = 2(i-1)n+n+2 - \sum_{j=2}^{i-1} (x_{j-}+x_{j+}-1) - N_i$.
This is true, since if any inequality is not saturated, say for $i$, then we can always add at least one $3_i$ symbol to the matching (gaining $B_i$ weight) and remove at most one $3_j$ symbol for each $j \in \{i+1,\ldots,k\}$ (losing less than $(k-1)B_{i+1} < B_i$ weight) and obtain a valid matching with larger cost, contradicting the optimality of our matching. 
Therefore, the number of matched $3_i$ symbols is precisely,
$$
x_{i-}+x_{i+} = 2(i-1)n+n+2 - \sum_{j=2}^{i-1} (x_{j-}+x_{j+}-1) - N_i.
$$

We can now formally analyze the tradeoff between (a) and (b), and prove that the optimal matching matches exactly $n$ intervals from each sequence.

\begin{claim}
In the optimal matching, $N_1=\cdots=N_k = n$.
\end{claim}  
\begin{proof}
Assume for contradiction that the claim does not hold, and we are in one of the two cases.

Case 1: For some $i \in [k]$, $N_i>n$.
In this case,  we consider any matching in which $N_i'=n$ intervals are matched in $P_i$, and in which the $x_{i-},x_{i+}$ values  are chosen optimally for all $i \in \{ 2,\ldots,k\}$.
Let $N_{m}=\min_{j=1}^k N_j$.
Clearly, the number of $3_m$ symbols in the new matching is at least $(x_{m-}+x_{m+}+(N_m-n))$, i.e. increased by $(N_m-n)$.  
Thus, in the contribution of part (b), we have gained a weight of at least $(N_m-n)B_m$.
To bound the loss in part (a), let $N_{min} = \min_{j=1}^k N_j$ and note that $N_m \leq n$.
The new contribution of part (a) is at least $n \cdot E_U$, while in the original matching, the contribution was at most $N_{min} \cdot E_U + (N_{m}-N_{min})\cdot (B-1)$.
Since $E_U> B$, the latter expression is maximized when $N_{min}$ is as large as possible, i.e. $N_{min}=n$, and we can upper bound it by $n \cdot E_U + (N_{m}-n)\cdot (B-1)$.
In total, the loss in part (a) is no more than $n \cdot E_U - n \cdot E_U + (N_{m}-n)\cdot (B-1) $ which is much less than $(N_m-n)B_m$, which is a contradiction to the optimality of our matching.

Case 2: For all $i \in [k]$, $N_i\leq n$, but for some $i \in [k]$, $N_i<n$.
In this case, we consider any matching in which $N_i'=n$ intervals are matched in $P_i$, and in which the $x_{i-},x_{i+}$ values  are chosen optimally for all $i \in \{ 2,\ldots,k\}$.
Clearly, for each $i \in \{ 2,\ldots,k\}$ the number of $3_i$ symbols in the new matching is at least $(x_{i-}+x_{i+}-i(n-N_i))$, i.e. decreased by no more than $i(n-N_i)$.  
Thus, in the contribution of part (b), we have lost a weight of at most $\sum_{i=2}^k i(n-N_i) B_i < k B_2\sum_{i=2}^k (n-N_i)$, but we have gained a larger weight, in part (a), as we show below.

Let $N_m = \min_{j=1}^k N_j$ and note that $\max_{j=1}^k N_i \leq n$. By Claim~\ref{cl:k1}, the part (a) contribution for the original matching had weight at most $N_m \cdot E_U + (n-N_m)\cdot (B-1)$, where $N_m \leq N_i$.
On the other hand, in the new matching, at least $n$ intervals are matched from each string, and therefore the contribution is at least $n\cdot E_U$.
Thus, in part (a) we gain at least $n\cdot E_U - N_m \cdot E_U - (n-N_m)\cdot (B-1)= (n-N_m)(E_U-B+1)$, which is larger than $k B_2\sum_{i=2}^k (n-N_i) \leq k B_2 (k-1) (n-N_m)$ since $E_U > C > k^2 B_2$.

\end{proof}

Finally, after we proved that $N_1=\cdots=N_k=n$, we know the exact contribution of both parts:
For part (b), by Claim~\ref{cl:tricky} and the optimality conditions on the $x_{i-},x_{i+}$ values, we get that $x_{2-}+x_{2+} = 2n+2$ and for $i \in \{2,\ldots,k\}$ we have $x_{i-}+x_{i+} = 2n+1$, and the total contribution is exactly $B_2 + (2n+1) \cdot \sum_{i=2}^k B_i$.
For part (a), by Claim~\ref{cl:k1}, the total contribution is $n\cdot E_U$.
Combined, the total score of our optimal matching is exactly $n\cdot E_U + B_2 + (2n+1) \cdot \sum_{i=2}^k B_i = E_G$.

\end{proof}

Note that the length of the sequences is $O(n \cdot d^{O(1)})$ while the largest weight used is $O(k^{O(k)} d^{O(1)})$ and thus Lemma~\ref{lem:kwlcs} implies the claimed bound.

\end{proof}

%


\section{Acknowledgments}
The authors thank Piotr Indyk for many useful discussions.
	\bibliographystyle{alpha}
	\bibliography{ref}

\end{document}